\documentclass[]{rAMF2e}
\usepackage{listings}
\usepackage[center]{caption}
\usepackage{hyperref}
\DeclareMathOperator{\Error}{Error}

\begin{document}
\doi{}
\issn{}  \issnp{}
\jvol{00} \jnum{00} \jyear{2015} %\jmonth{January--March}
\def\jobtag{}
\publisher{Unpublished}
\jname{}

\markboth{Fabien {Le Floc'h}}{Draft}

\title{More stochastic expansions for the pricing of vanilla options with cash dividends}
\author{Fabien {Le Floc'h}\thanks{{\em{Correspondence Address}}: 	Delft Institute of Applied Mathematics, TU Delft, Delft, The Netherlands. Email: \texttt{f.l.y.lefloch@tudelft.nl} \vspace{6pt}}}
\affil{Delft Institute of Applied Mathematics, TU Delft, Delft, The Netherlands.}
\date{\today}
\received{v1.1 released August 2019}

\maketitle
\newcommand{\sgn}{\mathop{\mathrm{sgn}}}
\begin{abstract}
	There is no exact closed form formula for pricing of European options with discrete cash dividends under the model where the underlying asset price follows a piecewise lognormal process with jumps at dividend ex-dates. This paper presents alternative expansions based on the technique of Etore and Gobet, leading to more robust first, second and third-order expansions across the range of strikes and the range of dividend dates.
\begin{keywords}Discrete dividends, cash dividends, stochastic expansion, option, pricing, Black-Scholes, finance\end{keywords}
\end{abstract}

\section{Introduction}
The Black-Scholes-Merton framework supposes a continuous dividend yield. In practice however, dividends are better modelled as fixed cash amounts, especially for short maturities (the first two or three years) where the uncertainty related to the dividend amount is low. For very short maturities (a quarter), the dividend amount is known exactly. For longer maturities, proportional dividends are more appropriate. From the arbitrage-free hypothesis and ignoring transaction costs, a discrete dividend implies that the stock jumps from the dividend amount at the dividend ex-date. 

In order to stay within the Black-Scholes-Merton framework, practitioners often compute the equivalent dividend yield, that is the dividend yield which leads to the same forward price as with pure discrete dividends for a given maturity. One issue however, is that the Black-Scholes-Merton model is then inconsistent around a dividend date. Indeed, around an ex-dividend date $t_i$, the option price must obey the following continuity relation \citep{bos2002finessing}
\begin{equation}\label{eqn:price_continuity}
V_{\textsf{call}}(S_0,K,t_{i}^-) = V_{\textsf{call}}(S_0, K-\delta_i,t_{i}^+)\,,
\end{equation}
for a cash dividend amount $\delta_i$, where $V_{\textsf{call}}(S_0,K,t)$ is the price of a European call option of strike $K$ and maturity $t$ and initial underlying price $S_0$.

As an example, let us consider the parameters $S_0=100, K=100, \delta_i = 1, t_{i} = 0.5$, volatility $\sigma=0.30$, interest rate $r = 0$, and apply the Black-Scholes formula using the constant dividend yield $\mu_T=r - \frac{1}{T}\ln\frac{S_0-\delta_i}{S_0}$, which preserves the forward price to maturity. We obtain
\begin{align*}
V_{\textsf{call}}(100,100,0.4999) = 8.446&\textmd{,}\quad V_{\textsf{call}}(100, 99,0.5001) = 8.363\textmd{.}
\end{align*}%graph footnote of price jump would be nice
The discontinuity is more obvious in the plot of the option price against the time to maturity (Figure \ref{fig:price_continuity01}).
\begin{figure}[h]
	\centering{\includegraphics[width=.48\textwidth]{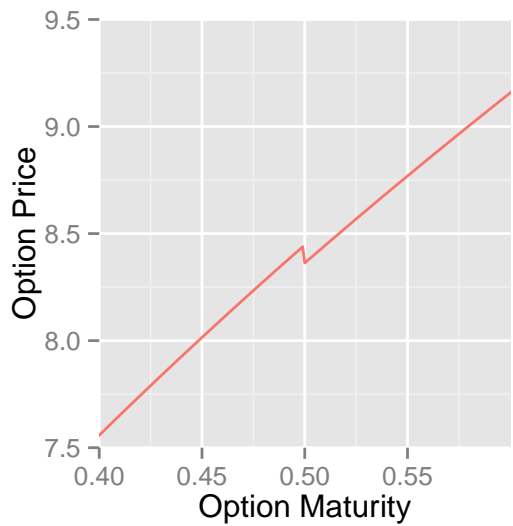}}
	\caption{Discontinuity of the option price under the forward model at a dividend date. $S_0=100, K=100, \delta_i = 1, t_{i} = 0.5, \sigma=30\%, r = q = 0$.}
	\label{fig:price_continuity01}
\end{figure}
As a consequence, under the Black-Scholes model, the volatility $\sigma$ must jump just after the dividend date in order for the prices to obey the continuity relationship. In our example, the call prices are equals if the volatility jumps to $\sigma(t_{i}^+) = 30.3\%$. Thus, in practice, the discontinuity may be solved by using a specific interpolation in time of the implied volatility surface around the dividend dates. It will also impact directly any implementation of the Dupire local volatility model, as the local volatility is computed from the time-derivative of the call option prices \citep{dupire1994pricing}. The impact on path-dependent options is more subtle, and there, the models can not be reconciled. In particular, American options are very sensitive to the dividend modelling: a dividend yield, a proportional dividend and a cash dividend lead to significantly different exercise boundaries, even when those are calibrated to give the same European option price \citep{meyer2002numerical,gottsche2011early,vellekoop2011integral}.
Other issues of the equivalent dividend yield approach are: an exploding forward dividend yield which will impact numerical accuracy when solving the Black-Scholes PDE, wrong delta and gamma hedging greeks.

In order to better take into account the discontinuity at dividend dates, and obtain a more meaningful hedging for European and American options, a natural model is to represent the equity or equity index as a piecewise lognormal process, which jumps from the dividend amount at each exercise date. There is however no exact analytical formula for this problem. Standard numerical techniques like Monte-Carlo, or better, finite difference or quadrature methods can be applied without much difficulty. The pricing of European options requires however fast evaluation methods in practice. While a quadrature might be fine for a single dividend \citep{haug2003back}, it is much slower for many dividends. As a consequence many have tried to develop corrections to the Black-Scholes formula to take into account the discrete dividends. %Bos and Vandermark
 \citet{bos2002finessing} proposed a simple adjustment to the spot and the strike. More recently more precise approximations were discovered \citep{sahel2011matching,zhang2011fast}. Of particularly interest is the stochastic Taylor expansion technique around a shifted lognormal model of % Etore and Gobet
  \citet{etore2012stochastic} leading to first-, second- and third-order (with regards to the dividend amount) formulas. Here, we will apply their Taylor expansion technique, but on a different proxy, aiming for more accurate approximations of European options prices under the piecewise lognormal model.

This paper is organized as follows. Section \ref{sec:notation} introduces the notation that will be used across the paper. In sections \ref{sec:expansion_forward}, \ref{sec:expansion_lehman}, we apply the Etore-Gobet expansion technique using respectively the forward model and the Lehman model of  %Bos and Vandermark 
\citet{bos2002finessing} as a proxy to obtain European option prices expansions of up to third-order. Finally, in section \ref{sec:expansion_numerical}, we compare the accuracy and performance against of our new expansions against other approximations for the price of European options under the piecewise-lognormal model.

\section{Notation}\label{sec:notation}
We borrow the notation of \citet{etore2012stochastic} as we reuse their technique. The asset price with proportional dividends $y_i \in [0,1)$ and cash dividends $\delta_i \geq 0$ is represented by the process $S^{(y,\delta)}$. The piecewise lognormal dynamics of the stock price under the risk-neutral measure $\mathbb{Q}$ between two dividend dates is:
\begin{equation}
dS_t^{(y,\delta)} = \sigma_t S_t^{(y,\delta)} dW_t + (r_t - q_t) S_t^{(y,\delta)}dt
\end{equation}
where $r_t$ is the risk free deterministic interest rate and $q_t$ a deterministic repo spread. And it jumps at the dividend dates $t_i$:
\begin{equation}
S_{t_i}^{(y,\delta)} = S_{t_i^-}^{(y,\delta)} - (\delta_i + y_i S_{t_i^-}^{(y,\delta)})=S_{t_i^-}^{(y,\delta)}(1-y_i)-\delta_i \textmd{.}
\end{equation}
The fictitious asset without dividends $S$ follows
\begin{equation}
dS_t = \sigma_t S_t dW_t + (r_t-q_t)S_t dt
\end{equation}
with initial value $S_0 = S_0^{(y,\delta)}$.
The zero discount factor $B_t$, the discount factor $D_t$ and the lognormal martingale $M_t$ are defined as
\begin{equation}
B_t = e^{-\int_0^t r_s ds} \textmd{, } D_t = e^{-\int_0^t (r_s-q_s)ds} \textmd{, } M_t = e^{\int_0^t \sigma_s dW_s - \frac{1}{2}\int_0^t \sigma_s^2 ds}\textmd{.}
\end{equation}
Thus, we have  $S_t = S_0 \frac{M_t}{D_t}$.
We set $\pi_{i,n} = \prod_{j=i+1}^{n} (1-y_j)$ and will make use of the simplified notation 
\begin{equation}
\hat{\delta}_i = \delta_i \pi_{i,n} \frac{D_{t_i}}{D_T}
\end{equation} 
where $T$ will typically be the option maturity.

Finally we recall the Lemma 1.1 of \citet{etore2012stochastic}:
\begin{lemma}\label{lemma_st}
	We have $ S_t^{(y,\delta)} = \pi_{0,n} S_t - \sum_{i=1}^n \delta_i \pi_{i,n} \frac{S_t}{S_{t_i}}$.
\end{lemma}
\section{Expansion around the forward}\label{sec:expansion_forward}
\citet{etore2012stochastic} consider an expansion around a shifted lognormal spot model where the asset price $ \bar{S}_t^{(y,\delta)}$ follows
\begin{equation}
\bar{S}_t^{(y,\delta)} = S_t - \mathbb{E}\left[\sum_{i=1}^n \delta_i \pi_{i,n} \frac{S_t}{S_{t_i}}\right] = \pi_{0,n}S_t - \sum_{i=1}^n \hat{\delta}_i
\end{equation}

We will instead do an expansion on the forward model :
\begin{equation}
F_t =  \left(\pi_{0,n}S_0 -\sum_{i=1}^n \delta_i \pi_{i,n}D_{t_i}\right)\frac{M_t}{D_t}
\end{equation}
We recall that the degree $k$ of smoothness of $h$ is defined as $H_k$: the function $h$ is $(k-1)$-times continuously differentiable and the $(k-1)$-th derivative is almost everywhere differentiable. Moreover, the derivatives are polynomially bounded: for some positive constant C and p, one has $|h(x)| + \sum_{j=1}^n |\partial_x^j h(x)| \leq C(1+|x|^p)$ for any $x \in \mathbb{R}$.

\subsection{First-order expansion around the forward}
\begin{theorem}\label{theorem_h_1}
For a smooth function $h$ satisfying $H_2$, we have
	\begin{align}
	\mathbb{E}\left[B_T h(S_T^{(y,\delta)} -K)\right] =&  	\mathbb{E}\left[B_T h(F_T -K)\right]\\
	&+ \sum_{i=1}^n	\hat{\delta}_i \left(\partial_K \mathbb{E}\left[B_T h(F_T e^{\int_{t_i}^T \sigma_s^2 ds}-K)\right]\right.\\
	&\left. - \partial_K \mathbb{E}\left[B_T  h(F_T e^{\int_{0}^T \sigma_s^2 ds}-K)\right]\right)\\
	&+ \Error_2(h)
	\end{align}
	where $|\Error_2(h)| \leq c(1+S_0^p)\sup_i\left(\delta_i \sigma \sqrt{t_i}\right)^2$.
\end{theorem}

\begin{proof}
	Let us first rewrite Lemma \ref{lemma_st} in terms of the forward $F_t$:
	\begin{align}
S_T^{(y,\delta)} &= \pi_{0,n}S_0 \frac{M_T}{D_T} - \sum_{i=1}^n \pi_{i,n}\delta_i \frac{M_T}{D_T}\frac{D_{t_i}}{M_{t_i}}\nonumber\\
&= F_t +  \sum_{i=1}^n \pi_{i,n}\delta_i \frac{M_T}{D_T} \left(D_{t_i}-\frac{D_{t_i}}{M_{t_i}}\right)\nonumber\\
&= F_t + \sum_{i=1}^n \hat{\delta}_i  \left(M_{T}-\frac{M_{T}}{M_{t_i}}\right)\textmd{.}
\end{align}
We then apply a simple Taylor expansion of order-1 on $h$:
\begin{align}
\mathbb{E}\left[B_T h(S_T^{(y,\delta)} -K)\right] =& \mathbb{E}\left[B_T h(F_T -K)\right]\nonumber\\
&+\sum_{i=1}^n \hat{\delta}_i \mathbb{E}\left[ B_T h'(F_T-K)  \left(M_{T}-\frac{M_{T}}{M_{t_i}}\right)\right]+ \Error_2(h)
\end{align}
where $|\Error_2(h)| \leq (1+S_0^p)\sup_i\left(\delta_i \|\frac{M_T}{M_{t_i}}-M_T \| \right)^2$. 
Reusing Lemma 3.3 of \citet{etore2012stochastic}, we have
\begin{equation}
\| \frac{M_T}{M_{t_i}}-M_T \| \leq \bar{\sigma} \sqrt{t_i} \textmd{.}
\end{equation}

The assumptions on $h$ allow us to interchange derivation and expectation:
\begin{align}
\mathbb{E}\left[ B_T h'(F_T-K)  \left(M_{T}-\frac{M_{T}}{M_{t_i}}\right)\right] =& -\partial_K\mathbb{E}\left[ B_T h(F_T-K)  M_{T}\right]+\partial_K\mathbb{E}\left[ B_T h(F_T-K)  \frac{M_{T}}{M_{t_i}}\right]
\end{align}
We now proceed to a change of measure defined by $\frac{M_T}{M_\tau}$ where $\tau=0$ and $\tau=t_i$ for the two expectations of the right on side. Under the new measure $\mathbb{Q}^\tau$, $\bar{W}_t = W_t - \int_0^t \sigma_s 1_{\tau \leq s \leq T} ds$ is a Brownian motion. $F_T$ under $\mathbb{Q}^\tau$ has the same law as $F_T e^{\int_\tau^T \sigma^2 ds}$ under $\mathbb{Q}$. Thus,
%(1+\sigma sqrt(T))= \sigma sqrt(t_i)
\begin{align}\label{change_1}
	\mathbb{E}\left[ B_T h(F_T-K)  M_{T}\right] &= \mathbb{E}\left[ B_T h(F_T e^{\int_{0}^T \sigma_s^2 ds}-K)\right]\\
 \mathbb{E}\left[ B_T h(F_T-K)  \frac{M_{T}}{M_{t_i}}\right] &= \mathbb{E}\left[ B_T hF_T e^{\int_{t_i}^T \sigma_s^2 ds} -K)  \right]
\end{align}\textmd{.}

	\end{proof}	
	
\subsection{Second-order expansion around the forward}
\begin{theorem}\label{theorem_h_2}
For a smooth function $h$ satisfying $H_3$, we have
\begin{align}
	&\mathbb{E}\left[B_T h(S_T^{(y,\delta)} -K)\right]\nonumber\\
	=&  	\mathbb{E}\left[B_T h(F_T -K)\right]\nonumber\\
	&+ \sum_{i=1}^n	\hat{\delta}_i \left(\partial_K \mathbb{E}\left[B_T h(F_T e^{\int_{t_i}^T \sigma_s^2 ds}-K)\right]\right.\left. - \partial_K \mathbb{E}\left[B_T  h(F_T e^{\int_{0}^T \sigma_s^2 ds}-K)\right]\right)\nonumber\\
	&+  \frac{1}{2} \sum_{1\leq i,j \leq n}	\hat{\delta}_i \hat{\delta}_j \partial^2_K \mathbb{E}\left[B_T h(F_T e^{\int_{t_i}^T \sigma_s^2 ds+\int_{t_j}^T \sigma_s^2 ds}-K)\right]e^{\int_{\max(t_i,t_j)}^T \sigma_s^2 ds}\nonumber\\
	&- \left(\sum_{j=1}^n	\hat{\delta}_j\right)  \sum_{i=1}^n \hat{\delta}_i \partial^2_K \mathbb{E}\left[B_T h(F_T e^{\int_{0}^T \sigma_s^2 ds+\int_{t_i}^T \sigma_s^2 ds}-K)\right]e^{\int_{t_i}^T \sigma_s^2 ds}\nonumber\\
	&+  \frac{1}{2}\left(\sum_{j=1}^n	\hat{\delta}_j\right)^2 \partial^2_K \mathbb{E}\left[B_T h(F_T e^{2\int_{0}^T \sigma_s^2 ds}-K)\right]e^{\int_{0}^T \sigma_s^2 ds}\nonumber\\
	&+ \Error_3(h)
\end{align}
where $|\Error_3(h)| \leq c(1+S_0^p)\sup_i\left(\delta_i \sigma \sqrt{t_i}\right)^3$.
\end{theorem}	
\begin{proof}
We start with a Taylor expansion of order-2 on $h$:
\begin{align}
	\mathbb{E}\left[B_T h(S_T^{(y,\delta)} -K)\right] =& \mathbb{E}\left[B_T h(F_T -K)\right]\nonumber\\
	&+\sum_{i=1}^n \hat{\delta}_i \mathbb{E}\left[ B_T h'(F_T-K)  \left(M_{T}-\frac{M_{T}}{M_{t_i}}\right)\right]\nonumber\\
	&+\frac{1}{2} \mathbb{E}\left[ B_T h''(F_T-K)  \left(\sum_{i=1}^n \hat{\delta}_i\left(M_{T}-\frac{M_{T}}{M_{t_i}}\right)\right)^2\right]\nonumber\\
	&+\Error_3(h)
\end{align}
where $|\Error_3(h)| \leq (1+S_0^p)\sup_i\left(\delta_i \|\frac{M_T}{M_{t_i}}-M_T \| \right)^3$. 

The first two terms stem from the first-order expansion of Theorem \ref{theorem_h_1}. The remaining term can be written as:
\begin{align}
	\mathbb{E}\left[ B_T h''(F_T-K)  \left(\sum_{i=1}^n \hat{\delta}_i\left(M_{T}-\frac{M_{T}}{M_{t_i}}\right)\right)^2\right] =& 	\sum_{1\leq i,j \leq n}	\hat{\delta}_i \hat{\delta}_j \mathbb{E}\left[ B_T h''(F_T-K) \frac{M_{T}}{M_{t_i}}\frac{M_{T}}{M_{t_j}}\right]\nonumber\\
	&-2\left(\sum_{j=1}^n \hat{\delta}_j\right) \sum_{i=1}^n \hat{\delta}_i \mathbb{E}\left[ B_T h''(F_T-K) \frac{M_{T}^2}{M_{t_i}}\right]\nonumber\\
	&+\left(\sum_{i=1}^n \hat{\delta}_i\right)^2 \mathbb{E}\left[ B_T h''(F_T-K) M_{T}^2\right]\textmd{.}
\end{align}

Similarly to equality \ref{change_1}, we proceed to a change of measure $\mathbb{Q}^{t_i,t_j}$ defined by $e^{\int_0^T\sigma_s (1_{t_i \leq s \leq T} + 1_{t_i \leq s \leq T})dW_s - \frac{1}{2}\int_{0}^T \sigma_s^2 (1_{t_i \leq s \leq T} + 1_{t_i \leq s \leq T})^2 ds }$ knowing that
\begin{equation}
\frac{M_{T}}{M_{t_i}}\frac{M_{T}}{M_{t_j}} = e^{\int_0^T\sigma_s (1_{t_i \leq s \leq T} + 1_{t_i \leq s \leq T})dW_s - \frac{1}{2}\int_{0}^T \sigma_s^2 (1_{t_i \leq s \leq T} + 1_{t_i \leq s \leq T})^2 ds} e^{\int_{\max(t_i,t_j)}^T \sigma_s^2 ds} \textmd{.}
\end{equation}
This means that $F_T$ under $\mathbb{Q}^{t_i,t_j}$ has the same law as $F_T e^{\int_{t_i}^T \sigma_s^2 ds+\int_{t_j}^T \sigma_s^2 ds}$ under $\mathbb{Q}$. Thus we have
\begin{align}
 \mathbb{E}\left[ B_T h''(F_T-K) \frac{M_{T}}{M_{t_i}}\frac{M_{T}}{M_{t_j}}\right]&=\mathbb{E}\left[ B_T h''(F_T e^{\int_{t_i}^T \sigma_s^2 ds+\int_{t_j}^T \sigma_s^2 ds} -K) \right]e^{\int_{\max(t_i,t_j)}^T \sigma_s^2 ds}\\
\mathbb{E}\left[ B_T h''(F_T-K) \frac{M_{T}^2}{M_{t_i}}\right] &=\mathbb{E}\left[ B_T h''(F_T e^{\int_{t_i}^T \sigma_s^2 ds+\int_{0}^T \sigma_s^2 ds} -K) \right]e^{\int_{t_i}^T \sigma_s^2 ds}\\
\mathbb{E}\left[ B_T h''(F_T-K) M_{T}^2\right]&= \mathbb{E}\left[ B_T h''(F_T e^{2\int_{0}^T \sigma_s^2 ds}-K) \right]e^{\int_{0}^T \sigma_s^2 ds}\textmd{.}
 \end{align}
Finally, the assumptions on $h$ allow us to interchange the derivative and the expectation, proving the theorem.
\end{proof}
Applying Theorem \ref{theorem_h_2} to $h(x)=  |x|^+$ would result in an approximation for the call option prices of second-order. Note that however this specific choice of $h$ does not obey $H_3$. Fortunately, \cite{etore2012stochastic} have proven that the results stay valid for such a function.

Under the forward model, the option price of strike $k$, forward $f = \mathbb{E}[F_T]$ and maturity $T$ is obtained by the Black formula \citep{black1976pricing}:
\begin{align}
	V_B(f,k) = \eta B_T\left[x \Phi(\eta d_1(f,k)) - k \Phi(\eta d_2(f,k)) \right]
\end{align}
where $\Phi$ is the cumulative normal distribution function, $\eta=1$ for a Call, $\eta=-1$ for a Put and 
\begin{align}
d_1(f,k) = \frac{1}{\sqrt{\int_0^T \sigma^2_s ds}}\ln\frac{f}{k}+\frac{1}{2}\sqrt{\int_0^T \sigma^2_s ds}\textmd{ , } d_2(f,k) = d_1(f,k) - \sqrt{\int_0^T \sigma^2_s ds}\textmd{.}
\end{align}
The derivatives towards the strike are:
\begin{align}
\partial_k V_B(f,k) &= -\eta B_T\Phi(\eta d_2(f,k)) \textmd{ ,}\\
\partial^2_k V_B(f,k) &= \frac{B_T}{k\sqrt{\int_0^T \sigma^2_s ds}}\phi(d_2(f,k))
\end{align}
where $\phi$ is the normal density function. The following relations will be useful to simplify the expansions of the vanilla option price:
\begin{align}
	&\partial_f V_B(f,k) = \eta B_T \Phi(\eta d_1(f,k)) \textmd{ ,}\\
	&\partial^2_f V_B(f,k) =\eta \frac{B_T}{f\sqrt{\int_0^T \sigma^2_s ds}} \phi(d_1(f,k)) \textmd{ ,}\\
	&d_2\left(f e^{\int_0^T \sigma^2_s ds},k\right) = d_1 \textmd{ ,}\\
	&\phi(d_2(f,k)) = \frac{f}{k}\phi(d_1(f,k))	\textmd{ ,}\\
	&\frac{1}{k} \phi\left( d_2\left(f e^{2\int_0^T \sigma^2_s ds - \int_0^t \sigma^2_s ds}, k\right) \right) e^{\int_t^T \sigma^2_s ds} = \frac{1}{f}\phi\left(d_1(f,k) - \frac{\int_0^t \sigma^2_s ds}{\sqrt{\int_0^T \sigma^2_s ds}}\right) \label{eqn:phid2s_d1} \textmd{ .}
\end{align}
Let $v_t = \sqrt{\int_0^t \sigma^2_s ds}$, applying Theorem \ref{theorem_h_2} to the vanilla option payoff leads to:
\begin{align}
	\mathbb{E}\left[B_T |\eta S_T^{(y,\delta)} - \eta K|^+\right]
	=& \eta B_T \left[f \Phi(\eta d_1) - K \Phi(\eta d_2)\right] \nonumber\\
	&-\eta B_T \sum_{i=1}^n	\hat{\delta}_i  \left( \Phi(\eta d_1 -\eta \frac{v_{t_i}^2}{v_T}) - \Phi(\eta d_1) \right) \nonumber\\
	&+  \frac{1}{2}B_T \sum_{1\leq i,j \leq n}	\hat{\delta}_i \hat{\delta}_j \frac{\phi(d_1 + v_T-\frac{v_{t_i}^2+v_{t_j}^2}{v_T})}{k v_T} e^{v_T^2 - v_{\max(t_i,t_j)}^2}\nonumber\\
	&- B_T \left(\sum_{j=1}^n	\hat{\delta}_j\right)  \sum_{i=1}^n \hat{\delta}_i \frac{\phi(d_1 - \frac{v_{t_i}^2}{v_T})}{f v_T}\nonumber\\
	&+  \frac{1}{2} B_T\left(\sum_{j=1}^n	\hat{\delta}_j\right)^2 \frac{\phi(d_1)}{f v_T}\nonumber\\
	&+ \Error_3(h)
\end{align}
where $|\Error_3(h)| \leq c(1+S_0^p)\sup_i\left(\delta_i \sigma \sqrt{t_i}\right)^3$, with $d_1 = d_1(f,K), d_2 = d_2(f,K)$ and $f=\pi_{0,n}\frac{S_0}{D_T} -\sum_{i=1}^n \delta_i \pi_{i,n}\frac{D_{t_i}}{D_T}$.

Note that it could also be expressed purely with derivatives of the vanilla option towards the forward (i.e. through the usual option Delta and Gamma greeks):
\begin{align}
	\mathbb{E}\left[B_T |\eta S_T^{(y,\delta)} - \eta K|^+\right] = &V_{B}\left(f,K\right)\nonumber\\
	&-\sum_{i=1}^{n} \hat{\alpha_i}\left[ \frac{\partial V_B}{\partial f}\left(f e^{-\int_{0}^{t_i}\sigma^2(s)ds},K\right)-\frac{\partial V_B}{\partial f}\left(f,K\right) 
	\right] \nonumber\\
	&+\frac{1}{2}\left[ \sum_{1 \leq i,j \leq n} \hat{\alpha_i} \hat{\alpha_j} e^{\int_{0}^{\min(t_i,t_j)}\sigma^2(s)ds} \frac{\partial^2 V_B}{\partial f^2}\left(f e^{-\int_{0}^{t_i}\sigma^2(s)ds-\int_{0}^{t_j}\sigma^2(s)ds},K\right) \right.\nonumber\\
	&\left. -2  \left(\sum_{j}^{n} \hat{\alpha_j}\right) \sum_{i=1}^{n} \hat{\alpha_i}\frac{\partial^2 V_B}{\partial f^2}\left(f e^{-\int_{0}^{t_i}\sigma^2(s)ds},K\right) \right.\nonumber\\
	&\left. + \left(\sum_{j}^{n} \hat{\alpha_j}\right)^2 \frac{\partial^2 V_B}{\partial f^2}\left(f ,K\right)\right]
\end{align}
where we used the equation (\ref{eqn:phid2s_d1}) and the identity $i+j-\max(i,j)=\min(i,j)$.

This is not too surprising. We performed the Taylor expansion on the strike variable as it is simpler to derive, but there is an equivalent formulation in terms of the forward variable.

\subsection{Third-order expansion around the forward}
\begin{theorem}\label{theorem_h_3}
	For a smooth function $h$ satisfying $H_4$, we have
	\begin{align}
		&\mathbb{E}\left[B_T h(S_T^{(y,\delta)} -K)\right]\nonumber\\
		=&  	\mathbb{E}\left[B_T h(F_T -K)\right]\nonumber\\
		&+ \sum_{i=1}^n	\hat{\delta}_i \left(\partial_K \mathbb{E}\left[B_T h(F_T e^{\int_{t_i}^T \sigma_s^2 ds}-K)\right]\right.\left. - \partial_K \mathbb{E}\left[B_T  h(F_T e^{\int_{0}^T \sigma_s^2 ds}-K)\right]\right)\nonumber\\
		&+  \frac{1}{2} \sum_{1\leq i,j \leq n}	\hat{\delta}_i \hat{\delta}_j \partial^2_K \mathbb{E}\left[B_T h(F_T e^{\int_{t_i}^T \sigma_s^2 ds+\int_{t_j}^T \sigma_s^2 ds}-K)\right]e^{\int_{\max(t_i,t_j)}^T \sigma_s^2 ds}\nonumber\\
		&- \left(\sum_{j=1}^n	\hat{\delta}_j\right)  \sum_{i=1}^n \hat{\delta}_i \partial^2_K \mathbb{E}\left[B_T h(F_T e^{\int_{0}^T \sigma_s^2 ds+\int_{t_i}^T \sigma_s^2 ds}-K)\right]e^{\int_{t_i}^T \sigma_s^2 ds}\nonumber\\
		&+  \frac{1}{2}\left(\sum_{j=1}^n	\hat{\delta}_j\right)^2 \partial^2_K \mathbb{E}\left[B_T h(F_T e^{2\int_{0}^T \sigma_s^2 ds}-K)\right]e^{\int_{0}^T \sigma_s^2 ds}\nonumber\\
		&+\frac{1}{6}\sum_{1\leq i,j,l \leq n}	\hat{\delta}_i \hat{\delta}_j \hat{\delta}_l \partial^3_K \mathbb{E}\left[B_T h(F_T e^{\int_{t_i}^T \sigma_s^2 ds+\int_{t_j}^T \sigma_s^2 ds+\int_{t_l}^T \sigma_s^2 ds}-K)\right]e^{\int_{\max(t_i,t_j)}^T \sigma_s^2 ds+\int_{\max(t_i,t_l)}^T \sigma_s^2 ds+\int_{\max(t_j,t_l)}^T \sigma_s^2 ds}\nonumber\\
		&-  \frac{1}{2}\left(\sum_{j=1}^n	\hat{\delta}_j\right) \sum_{1\leq i,j \leq n}	\hat{\delta}_i \hat{\delta}_j \partial^3_K \mathbb{E}\left[B_T h(F_T e^{\int_{t_i}^T \sigma_s^2 ds+\int_{t_j}^T \sigma_s^2 ds +\int_{0}^T \sigma_s^2}-K)\right]e^{\int_{\max(t_i,t_j)}^T \sigma_s^2 ds+\int_{t_i}^T \sigma_s^2 ds+\int_{t_j}^T \sigma_s^2 ds}\nonumber\\
		&+ \frac{1}{2}\left(\sum_{j=1}^n	\hat{\delta}_j\right)^2  \sum_{i=1}^n \hat{\delta}_i \partial^3_K \mathbb{E}\left[B_T h(F_T e^{2\int_{0}^T \sigma_s^2 ds+\int_{t_i}^T \sigma_s^2 ds}-K)\right]e^{2\int_{t_i}^T \sigma_s^2 ds+\int_{0}^T \sigma_s^2 ds}\nonumber\\
		&-\frac{1}{6}\left(\sum_{j=1}^n	\hat{\delta}_j\right)^3 \partial^3_K \mathbb{E}\left[B_T h(F_T e^{3\int_{0}^T \sigma_s^2 ds}-K)\right]e^{3\int_{0}^T \sigma_s^2 ds}\nonumber\\
		&+ \Error_4(h)
	\end{align}
	where $|\Error_4(h)| \leq c(1+S_0^p)\sup_i\left(\delta_i \sigma \sqrt{t_i}\right)^4$.
\end{theorem}	
\begin{proof}
The second-order terms follow directly from Theorem \ref{theorem_h_2}. The third-order terms come from the third-order Taylor expansion on $h$. The additional term from the Taylor expansion is:
\begin{align}
\frac{1}{6}\mathbb{E}\left[ B_T h'''(F_T-K)  \left(\sum_{i=1}^n \hat{\delta}_i\left(M_{T}-\frac{M_{T}}{M_{t_i}}\right)\right)^3\right]\text{.}
\end{align}
We can expand the cubic term as
\begin{align}
\left(\sum_{i=1}^n \hat{\delta}_i\left(M_{T}-\frac{M_{T}}{M_{t_i}}\right)\right)^3 
&=\left(\sum_{i=1}^n \hat{\delta}i \right)^3 M_T^3 - \left(\sum_{i=1}^n \hat{\delta}_i \frac{M_{T}}{M_{t_i}}\right)^3 \nonumber\\
& + 3 \left(\sum_{l=1}^n \hat{\delta}_l \right) \left(\sum_{1\leq i,j \leq n} \hat{\delta}_i \hat{\delta}_j \frac{M_{T}^3}{M_{t_i} M_{t_j}}\right) - 3 \left(\sum_{j=1}^n \hat{\delta}_j \right)^2 \left(\sum_{i=1}^n \hat{\delta}_i \frac{M_{T}^3}{M_{t_i}}\right) \text{.}
\end{align}
	We then compute the expectation of each term, using the linearity of the expectation and the identity
	\begin{align*}
	&\mathbb{E}\left[ B_T h'''(F_T-K) \frac{M_{T}^3}{M_{t_i} M_{t_j}M_{t_l}}\right] \\&=\partial_K^3 \mathbb{E}\left[ B_T h\left(F_T e^{\int_{t_i}^T \sigma^2_s ds +\int_{t_j}^T \sigma^2_s ds + \int_{t_l}^T \sigma^2_s ds}-K\right) \right]e^{\int_{\max(t_i,t_j)}^T \sigma^2_s ds + \int_{\max(t_i,t_l)}^T \sigma^2_s ds + \int_{\max(t_j,t_l)}^T \sigma^2_s ds}
	\end{align*}
	with $t_i = 0$ or $t_j=0$ or $t_i = T$ or $t_j=T$ depending on the expression to evaluate.
\end{proof}
\section{Expansion around the Lehman model}\label{sec:expansion_lehman}
The Lehman model consists in simple adjustment of the spot and the strike in the Black-Scholes formula. The spot is adjusted by the present value of near-term dividends, just as in the forward model, while the strike is adjusted by far-away dividends \citep{bos2002finessing}. Let the near and far parts be
\begin{align}
X^n_T &= \sum_i \frac{T-t_i}{T} \delta_i \pi_{i,n} \frac{D_{t_i}}{D_T}\textmd{ ,}\\
X^f_T &= \sum_i \frac{t_i}{T} \delta_i \pi_{i,n} \frac{D_{t_i}}{D_T}\textmd{ .}
\end{align}
In the Lehman model, the asset follows
\begin{equation}
\bar{S}_t = \left(\pi_{0,n} \frac{S_0}{D_t} - X^n_t \right) M_t - X^f_t\textmd{ .}
\end{equation}
The vanilla option price can be obtained through the Black formula $V_{B}(F, K)$ for a forward $F$, strike $K$, expiry $T$:
\begin{equation}
V(S_0,K) = V_{B}\left(\frac{S_0}{D_T} \pi_{0,n} - X^n_T, K + X^f_T\right)\textmd{ .}
\end{equation}

Let us define
\begin{align}
	\bar{F}_t &= \left(\pi_{0,n} \frac{S_0}{D_t} - X^n_t \right) M_t \textmd{ ,}\\
	\bar{f} &= \frac{S_0}{D_T} \pi_{0,n} - X^n_T \textmd{ ,}\\
	\bar{k} &=  K + X^f_T \textmd{ .}
\end{align}

\subsection{First-order expansion around the Lehman model}
\begin{theorem}\label{theorem_hl_1}
	For a smooth function $h$ satisfying $H_2$, we have
	\begin{align}
		&\mathbb{E}\left[B_T h(S_T^{(y,\delta)} -K)\right] \\
		=& 	\mathbb{E}\left[B_T h\left(\bar{F}_T - \bar{k} \right)\right]\nonumber\\
		&+ \sum_{i=1}^n \frac{T-t_i}{T}	\hat{\delta}_i \left(\partial_K \mathbb{E}\left[B_T h\left(\bar{F}_T e^{\int_{t_i}^T \sigma_s^2 ds}-\bar{k}\right)\right] - \partial_K \mathbb{E}\left[B_T  h\left(\bar{F}_T e^{\int_{0}^T \sigma_s^2 ds}-\bar{k}\right)\right]\right)\nonumber\\
			&+ \sum_{i=1}^n \frac{t_i}{T}	\hat{\delta}_i \left(\partial_K \mathbb{E}\left[B_T h\left(\bar{F}_T e^{\int_{t_i}^T \sigma_s^2 ds}-\bar{k}\right)\right] - \partial_K \mathbb{E}\left[B_T  h\left(\bar{F}_T-\bar{k}\right)\right]\right)		+ \Error_2(h)
	\end{align}
	where $|\Error_2(h)| \leq c(1+S_0^p)\sup_i\left( \left(\delta_i\sigma\frac{T-t_i}{T} \sqrt{t_i}\right)^2 + \left(  \delta_i\sigma\frac{t_i}{T}\sqrt{T-t_i}\right)^2\right)$.
\end{theorem}

\begin{proof}
	Let us first rewrite Lemma \ref{lemma_st} in terms of $\bar{S}_T$:
	\begin{align}
		S_T^{(y,\delta)} &= \pi_{0,n}S_0 \frac{M_T}{D_T} - \sum_{i=1}^n \pi_{i,n}\delta_i \frac{M_T}{D_T}\frac{D_{t_i}}{M_{t_i}}\nonumber\\
		&= \bar{S}_T + X^n_T M_T + X^f_T - \sum_{i=1}^n \frac{T-t_i+t_i}{T}\hat{\delta}_i  \frac{M_T}{M_{t_i}}\nonumber\\
		&= \bar{S}_T + \sum_{i=1}^n \frac{T-t_i}{T} \hat{\delta}_i \left( M_T-\frac{M_T}{M_{t_i}} \right) + \sum_{i=1}^n \frac{t_i}{T}\hat{\delta}_i \left(1- \frac{M_T}{M_{t_i}} \right) \textmd{.}
	\end{align}
	As for Theorem \ref{theorem_h_1}, we apply a simple Taylor expansion of order-1 on $h$:
	\begin{align}
		\mathbb{E}\left[B_T h(S_T^{(y,\delta)} -K)\right] =& \mathbb{E}\left[B_T h(\bar{F}_T -\bar{k})\right]\nonumber\\
		&+ \sum_{i=1}^n \frac{T-t_i}{T} \hat{\delta}_i \mathbb{E}\left[ B_T h'(\bar{F}_T-\bar{k})  \left(M_{T}-\frac{M_{T}}{M_{t_i}}\right)\right]\nonumber\\
		&+ \sum_{i=1}^n \frac{t_i}{T} \hat{\delta}_i \mathbb{E}\left[ B_T h'(\bar{F}_T-\bar{k})  \left(1-\frac{M_{T}}{M_{t_i}}\right)\right]+ \Error_2(h)
	\end{align}
	where $|\Error_2(h)| \leq (1+S_0^p)\sup_i\left(\delta_i\frac{T-t_i}{T} \|\frac{M_T}{M_{t_i}}-M_T \| + \delta_i\frac{t_i}{T} \|\frac{M_T}{M_{t_i}}-1 \| \right)^2$. 
	
The first part can be obtained through the same steps as in Theorem \ref{theorem_h_1} applied to $\delta^n_i = \frac{T-t_i}{T}\delta_i $ and the second part can be obtained through the same steps as in Theorem 2.1 of \citet{etore2012stochastic} applied to $\delta^f_i = \frac{t_i}{T}\delta_i $ .
	
	\end{proof}
\subsection{Second-order expansion around the Lehman model}
\begin{theorem}\label{theorem_hl_2}
	For a smooth function $h$ satisfying $H_3$, we have
	\begin{align}
		&\mathbb{E}\left[B_T h(S_T^{(y,\delta)} -K)\right] \\
		=& 	\mathbb{E}\left[B_T h\left(\bar{F}_T - \bar{k} \right)\right]\nonumber\\
		&+ \sum_{i=1}^n \frac{T-t_i}{T}	\hat{\delta}_i \left(\partial_K \mathbb{E}\left[B_T h\left(\bar{F}_T e^{\int_{t_i}^T \sigma_s^2 ds}-\bar{k}\right)\right] - \partial_K \mathbb{E}\left[B_T  h\left(\bar{F}_T e^{\int_{0}^T \sigma_s^2 ds}-\bar{k}\right)\right]\right)\nonumber\\
		&+ \sum_{i=1}^n \frac{t_i}{T}	\hat{\delta}_i \left(\partial_K \mathbb{E}\left[B_T h\left(\bar{F}_T e^{\int_{t_i}^T \sigma_s^2 ds}-\bar{k}\right)\right] - \partial_K \mathbb{E}\left[B_T  h\left(\bar{F}_T-\bar{k}\right)\right]\right)\nonumber\\
		&+  \frac{1}{2} \sum_{1\leq i,j \leq n}	\frac{T+t_j-t_i}{T}\hat{\delta}_i \hat{\delta}_j \partial^2_K \mathbb{E}\left[B_T h(\bar{F}_T e^{\int_{t_i}^T \sigma_s^2 ds+\int_{t_j}^T \sigma_s^2 ds}-\bar{k})\right]e^{\int_{\max(t_i,t_j)}^T \sigma_s^2 ds}\nonumber\\
		&- \left(\sum_{j=1}^n \frac{T-t_j}{T}	\hat{\delta}_j\right)  \sum_{i=1}^n \hat{\delta}_i \partial^2_K \mathbb{E}\left[B_T h(\bar{F}_T e^{\int_{0}^T \sigma_s^2 ds+\int_{t_i}^T \sigma_s^2 ds}-\bar{k})\right]e^{\int_{t_i}^T \sigma_s^2 ds}\nonumber\\
		&- \left(\sum_{j=1}^n \frac{t_j}{T}	\hat{\delta}_j\right)  \sum_{i=1}^n  \hat{\delta}_i \partial^2_K \mathbb{E}\left[B_T h(\bar{F}_T e^{\int_{t_i}^T \sigma_s^2 ds}-\bar{k})\right]\nonumber\\
		&+  \frac{1}{2}\left(\sum_{j=1}^n  \frac{T-t_j}{T}	\hat{\delta}_j\right)^2 \partial^2_K \mathbb{E}\left[B_T h(\bar{F}_T e^{2\int_{0}^T \sigma_s^2 ds}-\bar{k})\right]e^{\int_{0}^T \sigma_s^2 ds}\nonumber\\
		&+ \frac{1}{2}\left(\sum_{j=1}^n  \frac{t_j}{T}	\hat{\delta}_j\right)^2\partial^2_K \mathbb{E}\left[B_T h(\bar{F}_T-\bar{k})\right]\nonumber\\
		&+ \left(\sum_{j=1}^n  \frac{t_j}{T}	\hat{\delta}_j\right)\left(\sum_{j=1}^n  \frac{T-t_j}{T}	\hat{\delta}_j\right)\partial^2_K \mathbb{E}\left[B_T h(\bar{F}_T e^{\int_{0}^T \sigma_s^2 ds}-\bar{k})\right] + \Error_3(h)
	\end{align}
	where $|\Error_3(h)| \leq c(1+S_0^p)\sup_i\left( \left(\delta_i\sigma\frac{T-t_i}{T} \sqrt{t_i}\right)^3 + \left(  \delta_i\sigma\frac{t_i}{T}\sqrt{T-t_i}\right)^3\right)$.
\end{theorem}

\begin{proof}
Let  $\hat{\delta}^n_i = \frac{T-t_i}{T}\hat{\delta}_i, \hat{\delta}^f_i = \frac{t_i}{T}\hat{\delta}_i $
We start with a Taylor expansion of order-2 on $h$:
\begin{align}
	\mathbb{E}\left[B_T h(S_T^{(y,\delta)} -K)\right] =& \mathbb{E}\left[B_T h(\bar{F}_T - \bar{k})\right]\nonumber\\
	&+\sum_{i=1}^n \hat{\delta}^n_i \mathbb{E}\left[ B_T h'(\bar{F}_T-\bar{k})  \left(M_{T}-\frac{M_{T}}{M_{t_i}}\right)\right]\nonumber\\
	&+\sum_{i=1}^n \hat{\delta}^f_i \mathbb{E}\left[ B_T h'(\bar{F}_T-\bar{k})  \left(1-\frac{M_{T}}{M_{t_i}}\right)\right]\nonumber\\
	&+\frac{1}{2} \mathbb{E}\left[ B_T h''(\bar{F}_T-\bar{k})  \left(\sum_{i=1}^n \hat{\delta}^n_i\left(M_{T}-\frac{M_{T}}{M_{t_i}}\right)+ \hat{\delta}^f_i\left(1-\frac{M_{T}}{M_{t_i}}\right)\right)^2\right]\nonumber\\
	&+\Error_3(h)
\end{align}
where $|\Error_3(h)| \leq (1+S_0^p)\sup_i\left(\delta_i \|\frac{M_T}{M_{t_i}}-M_T \| \right)^3$. 

The first three terms stem from the first-order expansion of Theorem \ref{theorem_hl_1}. The remaining term can be written as:
\begin{align}
	&\mathbb{E}\left[ B_T h''(\bar{F}_T-\bar{k})  \left(\sum_{i=1}^n \hat{\delta}^n_i\left(M_{T}-\frac{M_{T}}{M_{t_i}}\right)+ \hat{\delta}^f_i\left(1-\frac{M_{T}}{M_{t_i}}\right)\right)^2\right]\nonumber\\
	=&\mathbb{E}\left[ B_T h''(\bar{F}_T-\bar{k})  \left(\sum_{i=1}^n \hat{\delta}^n_i\left(M_{T}-\frac{M_{T}}{M_{t_i}}\right)\right)^2\right] + \mathbb{E}\left[ B_T h''(\bar{F}_T-\bar{k})  \left(\sum_{i=1}^n \hat{\delta}^f_i\left(1-\frac{M_{T}}{M_{t_i}}\right)\right)^2\right]\nonumber\\
	&+ 2 \mathbb{E}\left[ B_T h''(\bar{F}_T-\bar{k})  \left(\sum_{i=1}^n \hat{\delta}^n_i\left(M_{T}-\frac{M_{T}}{M_{t_i}}\right)\right) \left(\sum_{j=1}^n \hat{\delta}^f_j\left(1-\frac{M_{T}}{M_{t_j}}\right)\right)\right]
	\end{align}
The first two terms have been computed in the proof of Theorem \ref{theorem_h_2} above and in the proof of Theorem 2.2 of \cite{etore2012stochastic}. The remaining term can be written as:

	\begin{align}
		& \mathbb{E}\left[ B_T h''(\bar{F}_T-\bar{k})  \left(\sum_{i=1}^n \hat{\delta}^n_i\left(M_{T}-\frac{M_{T}}{M_{t_i}}\right)\right) \left(\sum_{j=1}^n \hat{\delta}^f_j\left(1-\frac{M_{T}}{M_{t_j}}\right)\right)\right]\\
	 =& 	\sum_{1\leq i,j \leq n}	\hat{\delta}^n_i \hat{\delta}^f_j \mathbb{E}\left[ B_T h''(\bar{F}_T-\bar{k}) \frac{M_{T}}{M_{t_i}}\frac{M_{T}}{M_{t_j}}\right]\nonumber\\
	&-\left(\sum_{j=1}^n \hat{\delta}^n_j\right) \sum_{i=1}^n \hat{\delta}_i^f \mathbb{E}\left[ B_T h''(\bar{F}_T-\bar{k}) \frac{M_{T}^2}{M_{t_i}}\right]-\left(\sum_{j=1}^n \hat{\delta}^f_j\right) \sum_{i=1}^n \hat{\delta}_i^n \mathbb{E}\left[ B_T h''(\bar{F}_T-\bar{k}) \frac{M_{T}}{M_{t_i}}\right]\nonumber\\
	&+\left(\sum_{i=1}^n \hat{\delta}_i^n\right)\left(\sum_{i=1}^n \hat{\delta}_i^f\right) \mathbb{E}\left[ B_T h''(\bar{F}_T-\bar{k}) M_{T}\right]\textmd{.}
\end{align}

The terms in $\hat{\delta}_i \hat{\delta}_j$ can be assembled together to give:
\begin{align*}
	&\sum_{1\leq i,j \leq n}	\hat{\delta}^n_i \hat{\delta}^n_j \mathbb{E}\left[ B_T h''(\bar{F}_T-\bar{k}) \frac{M_{T}}{M_{t_i}}\frac{M_{T}}{M_{t_j}}\right]	+	\sum_{1\leq i,j \leq n}	\hat{\delta}^f_i \hat{\delta}^f_j \mathbb{E}\left[ B_T h''(\bar{F}_T-\bar{k}) \frac{M_{T}}{M_{t_i}}\frac{M_{T}}{M_{t_j}}\right]\\
	&+
	2\sum_{1\leq i,j \leq n}	\hat{\delta}^n_i \hat{\delta}^f_j \mathbb{E}\left[ B_T h''(\bar{F}_T-\bar{k}) \frac{M_{T}}{M_{t_i}}\frac{M_{T}}{M_{t_j}}\right]\\
	=& \sum_{1\leq i,j \leq n} \left(1 +\frac{t_j-t_i}{T}\right)	\hat{\delta}_i \hat{\delta}_j \mathbb{E}\left[ B_T h''(\bar{F}_T-\bar{k}) \frac{M_{T}}{M_{t_i}}\frac{M_{T}}{M_{t_j}}\right]
\end{align*}
as $\left(1-\frac{t_i}{T}\right) \left(1-\frac{t_j}{T}\right) + \frac{t_j}{T}\frac{t_i}{T}+2\left(1-\frac{t_i}{T}\right)\frac{t_j}{T} = 1 +\frac{t_j}{T}- \frac{t_i}{T}$.

Similarly to Theorem \ref{theorem_h_2}, changes of measure and interchange of expectation and derivatives will lead to the proof.
\end{proof}

Although the formula of Theorem \ref{theorem_hl_2} looks long, it will not be more costly to evaluate than the simpler formula of Theorem \ref{theorem_h_2} when applied to vanilla options. Let $v_t = \sqrt{\int_0^t \sigma^2_s ds}$, applying Theorem \ref{theorem_hl_2} to the vanilla option payoff leads to:
\begin{align}
	\mathbb{E}\left[B_T |\eta S_T^{(y,\delta)} - \eta K|^+\right]
	=& \eta B_T \left[\bar{f} \Phi(\eta d_1) - \bar{k} \Phi(\eta d_2)\right] \nonumber\\
	&-\eta B_T \sum_{i=1}^n	\hat{\delta}_i  \left( \Phi(\eta d_1 -\eta \frac{v_{t_i}^2}{v_T}) - \frac{T-t_i}{T}\Phi(\eta d_1) - \frac{t_i}{T}\Phi(\eta d_2) \right) \nonumber\\
	&+  \frac{1}{2}B_T \sum_{1\leq i,j \leq n}	\hat{\delta}_i \hat{\delta}_j \frac{T+t_j-t_i}{T} \frac{\phi(d_1 + v_T-\frac{v_{t_i}^2+v_{t_j}^2}{v_T})}{\bar{k} v_T} e^{v_T^2 - v_{\max(t_i,t_j)}^2}\nonumber\\
	&- B_T \left(\sum_{j=1}^n \frac{T-t_j}{T}\hat{\delta}_j\right) \sum_{i=1}^n \hat{\delta_i} \frac{\phi(d_1 - \frac{v_{t_i}^2}{v_T})}{\bar{f} v_T}  \nonumber - B_T \left(\sum_{j=1}^n \frac{t_j}{T}\hat{\delta}_j\right) \sum_{i=1}^n \hat{\delta_i} \frac{\phi(d_1 -\frac{v_{t_i}^2}{v_T})}{\bar{k} v_T} \nonumber \\
	&+ \frac{1}{2} B_T \left[ \left(\sum_{j=1}^n \frac{T-t_j}{T}	\hat{\delta}_j\right)^2 \frac{\phi(d_1)}{\bar{f} v_T} + \left(\sum_{j=1}^n \frac{t_j}{T}	\hat{\delta}_j\right)^2 \frac{\phi(d_2)}{\bar{k} v_T} \right]\nonumber\\
		&+ B_T \left(\sum_{j=1}^n \frac{T-t_j}{T}	\hat{\delta}_j\right) \left(\sum_{j=1}^n \frac{t_j}{T}	\hat{\delta}_j\right) \frac{\phi(d_1)}{\bar{k} v_T} \nonumber\\
	&+ \Error_3(h)
\end{align}
where $|\Error_3(h)| \leq c(1+S_0^p)\sup_i\left( \left(\delta_i\sigma\frac{T-t_i}{T} \sqrt{t_i}\right)^3 + \left(  \delta_i\sigma\frac{t_i}{T}\sqrt{T-t_i}\right)^3\right)$, with $d_1 = d_1(\bar{f},\bar{k}), d_2 = d_2(\bar{f},\bar{k})$.

The double sum can be split into two symmetric parts, leading to a nearly 50\% performance gain over a naive implementation when the number of dividends is high as we have for any function $u$:

\begin{align}
&\sum_{1\leq i,j \leq n} (T+t_j-t_i) \hat\delta_i \hat\delta_j  u(v_{t_i}^2+v_{t_j}^2,\max(t_i,t_j)) \nonumber\\
= & \sum_{i=1}^n T \hat\delta_i^2  u(2 v_{t_i}^2,t_i) +  \sum_{i=1}^n   \sum_{j=1}^{i-1}(T+t_j-t_i)  \hat\delta_i \hat\delta_j u(v_{t_i}^2+v_{t_j}^2,t_i) \nonumber\\
&+ \sum_{i=1}^n \sum_{j=i+1}^{n} (T+t_j-t_i) \hat\delta_i \hat\delta_j u(v_{t_i}^2+v_{t_j}^2, t_j)\nonumber\\
= & \sum_{i=1}^n T \hat\delta_i^2  u(2v_{t_i}^2,t_i) +  \sum_{i=1}^n   \sum_{j=i+1}^{n}(T-t_j+t_i)  \hat\delta_i \hat\delta_j u(v_{t_i}^2+v_{t_j}^2, t_j)\nonumber\\
&+ \sum_{i=1}^n \sum_{j=i+1}^{n} (T+t_j-t_i) \hat\delta_i \hat\delta_j u(v_{t_i}^2+v_{t_j}^2,t_j)\nonumber\\
= & \sum_{i=1}^n T \hat\delta_i^2  u(2v_{t_i}^2,t_i) +  2\sum_{i=1}^n   \sum_{j=i+1}^{n}T  \hat\delta_i \hat\delta_j u(v_{t_i}^2+v_{t_j}^2,t_i) \textmd{.}
\end{align}
%sum_i sum_j  A(i)B(j)  = sum_k (A((k+l)/2)B((k-l)/2))
%k = i+j , 0, 2n   l = i-j   
% 1 1   2 1      sum(i=1,n-1) sum(j>i) w_i w_j h(i+j, max(i+j))
% 1 2   3 2  a    + sum(i=2,n) sum(j<i)
% 1 3   4 3  a    = sum(i=1,n-1) sum(j>i,n)
% 2 1   3 2  b        + sum()
% 2 2   4 2
% 2 3   5 3   a
% 3 1   4 3
% 3 2   5 3
In terms of computational cost, this expansion is not more costly than the original second-order expansion of \citet{etore2012stochastic}: in a careful implementation, counting in the normal density function evaluations, it involves the same number of exponential functions evaluations.

In theory, we can expect the expansion around the displaced strike to be more accurate when the dividends are near the maturity date of the option: in the extreme case where the dividends happen at maturity, the zero-th order price stemming from the shifted Black formula is exact. Similarly, we can expect the expansion around the forward to be more accurate when the dividends are near the valuation date. In this regard, being a mix of the two, the expansion around the Lehman model should provide a more stable approximation over the full range of dividend dates.

\subsection{Third-order expansion around the Lehman model}
\begin{theorem}\label{theorem_hl_3}
	For a smooth function $h$ satisfying $H_4$, the third-order correction is:
\begin{align}
& \frac{1}{6} \mathbb{E}\left[ B_T h'''(\bar{F}_T-\bar{K})  \left(\sum_{i=1}^n \hat{\delta}^n_i\left(M_{T}-\frac{M_{T}}{M_{t_i}}\right)+\hat{\delta}^f_j\left(1-\frac{M_{T}}{M_{t_j}}\right) \right)^3\right]\\
&=\frac{1}{6}\left(\sum_{j=1}^n \hat{\delta}^f_i \right)^3\partial_K^3 \mathbb{E}\left[ B_T h\left(\bar{F}_T-\bar{K}\right) \right]\nonumber\\
&+\frac{1}{6}\left(\sum_{i=1}^n \hat{\delta}^n_i \right)^3 \partial_K^3 \mathbb{E}\left[ B_T h\left(\bar{F}_T e^{3\int_{0}^T \sigma^2_s ds}-\bar{K}\right) \right]e^{3\int_{0}^T \sigma^2_s ds}\nonumber\\
&-\frac{1}{6}\sum_{1\leq i,j,l \leq n} \hat{\delta}_i \hat{\delta}_j \hat{\delta}_l \partial_K^3 \mathbb{E}\left[ B_T h\left(\bar{F}_T e^{\int_{t_i}^T \sigma^2_s ds +\int_{t_j}^T \sigma^2_s ds + \int_{t_l}^T \sigma^2_s ds}-\bar{K}\right) \right]e^{\int_{\max(t_i,t_j)}^T \sigma^2_s ds + \int_{\max(t_i,t_l)}^T \sigma^2_s ds + \int_{\max(t_j,t_l)}^T \sigma^2_s ds}\nonumber\\
&+\frac{1}{2}\left(\sum_{i=1}^n \hat{\delta}^n_i \right) \left(\sum_{j=1}^n \hat{\delta}^f_j \right)^2 \partial_K^3 \mathbb{E}\left[ B_T h\left(\bar{F}_Te^{\int_{0}^T \sigma^2_s ds}-\bar{K}\right) \right]\nonumber\\
&+\frac{1}{2}\left(\sum_{i=1}^n \hat{\delta}^n_i \right)^2 \left(\sum_{j=1}^n \hat{\delta}^f_j \right) \partial_K^3 \mathbb{E}\left[ B_T h\left(\bar{F}_Te^{2\int_{0}^T \sigma^2_s ds}-\bar{K}\right) \right]e^{\int_{0}^T \sigma^2_s ds}\nonumber\\
&+\frac{1}{2}\left(\sum_{l=1}^n \hat{\delta}^n_l \right) \sum_{1\leq i,j \leq n} \hat{\delta}_i \hat{\delta}_j \partial_K^3 \mathbb{E}\left[ B_T h\left(\bar{F}_T e^{\int_{t_i}^T \sigma^2_s ds +\int_{t_j}^T \sigma^2_s ds + \int_{0}^T \sigma^2_s ds}-\bar{K}\right) \right]e^{\int_{\max(t_i,t_j)}^T \sigma^2_s ds + \int_{t_i}^T \sigma^2_s ds + \int_{t_j}^T \sigma^2_s ds}\nonumber\\
&+\frac{1}{2}\left(\sum_{l=1}^n \hat{\delta}^f_l \right) \sum_{1\leq i,j \leq n} \hat{\delta}_i \hat{\delta}_j \partial_K^3 \mathbb{E}\left[ B_T h\left(\bar{F}_T e^{\int_{t_i}^T \sigma^2_s ds +\int_{t_j}^T \sigma^2_s ds}-\bar{K}\right) \right]e^{\int_{\max(t_i,t_j)}^T \sigma^2_s ds }\nonumber\\
&-\frac{1}{2}\left(\sum_{l=1}^n \hat{\delta}^n_l \right)^2 \sum_{i=1}^n \hat{\delta}_i \partial_K^3 \mathbb{E}\left[ B_T h\left(\bar{F}_T e^{\int_{t_i}^T \sigma^2_s ds +2 \int_{0}^T \sigma^2_s ds}-\bar{K}\right) \right]e^{2\int_{t_i}^T \sigma^2_s ds + \int_{0}^T \sigma^2_s ds}\nonumber\\
&-\frac{1}{2}\left(\sum_{l=1}^n \hat{\delta}^f_l \right)^2 \sum_{i=1}^n \hat{\delta}_i \partial_K^3 \mathbb{E}\left[ B_T h\left(\bar{F}_T e^{\int_{t_i}^T \sigma^2_s ds}-\bar{K}\right) \right]\nonumber\\
&-\left(\sum_{j=1}^n \hat{\delta}^n_j \right) \left(\sum_{l=1}^n \hat{\delta}^f_l \right) \sum_{i=1}^n \hat{\delta}_i \partial_K^3 \mathbb{E}\left[ B_T h\left(\bar{F}_T e^{\int_{t_i}^T \sigma^2_s ds + \int_{0}^T \sigma^2_s ds}-\bar{K}\right) \right]e^{\int_{t_i}^T \sigma^2_s ds}
\end{align}
\end{theorem}
\begin{proof}
We start by expanding the cubic term as
\begin{align}
&	\left(\sum_{i=1}^n \hat{\delta}^n_i\left(M_{T}-\frac{M_{T}}{M_{t_i}}\right)+\hat{\delta}^f_i\left(1-\frac{M_{T}}{M_{t_j}}\right) \right)^3 = \left(\sum_{i=1}^n \hat{\delta}^n_i M_{T} +\hat{\delta}^f_i - \left( \hat{\delta}^n_i+ \hat{\delta}^f_i\right)\frac{M_{T}}{M_{t_j}}\right)^3\nonumber\\
&=\left(\sum_{i=1}^n \hat{\delta}^n_i \right)^3 M_T^3 + \left(\sum_{i=1}^n \hat{\delta}^f_i \right)^3 - \left(\sum_{i=1}^n \hat{\delta}_i \frac{M_{T}}{M_{t_i}}\right)^3\nonumber \\
&+3\left(\sum_{i=1}^n \hat{\delta}^n_i \right)^2 \left(\sum_{j=1}^n \hat{\delta}^f_j \right) M_T^2 + 3\left(\sum_{i=1}^n \hat{\delta}^n_i \right) \left(\sum_{j=1}^n \hat{\delta}^f_j \right)^2 M_T\nonumber   \\
& + 3 \left(\sum_{l=1}^n \hat{\delta}^n_l \right) \left(\sum_{1\leq i,j \leq n} \hat{\delta}_i \hat{\delta}_j \frac{M_{T}^3}{M_{t_i} M_{t_j}}\right) + 3 \left(\sum_{l=1}^n \hat{\delta}^f_l \right) \left(\sum_{1\leq i,j \leq n} \hat{\delta}_i \hat{\delta}_j \frac{M_{T}^2}{M_{t_i} M_{t_j}}\right)\nonumber\\
& - 3 \left(\sum_{j=1}^n \hat{\delta}^n_j \right)^2 \left(\sum_{i=1}^n \hat{\delta}_i \frac{M_{T}^3}{M_{t_i}}\right) - 3 \left(\sum_{j=1}^n \hat{\delta}^f_j \right)^2 \left(\sum_{i=1}^n \hat{\delta}_i \frac{M_{T}}{M_{t_i}}\right)- 6 \left(\sum_{j=1}^n \hat{\delta}^n_j \right)\left(\sum_{l=1}^n \hat{\delta}^f_l \right)\left(\sum_{i=1}^n \hat{\delta}_i \frac{M_{T}^2}{M_{t_i}}\right)
	\end{align}
We then compute the expectation of each term, using the linearity of the expectation and the identity
\begin{align*}
	 &\mathbb{E}\left[ B_T h'''(F_T-K) \frac{M_{T}^3}{M_{t_i} M_{t_j}M_{t_l}}\right] \\&=\partial_K^3 \mathbb{E}\left[ B_T h\left(F_T e^{\int_{t_i}^T \sigma^2_s ds +\int_{t_j}^T \sigma^2_s ds + \int_{t_l}^T \sigma^2_s ds}-K\right) \right]e^{\int_{\max(t_i,t_j)}^T \sigma^2_s ds + \int_{\max(t_i,t_l)}^T \sigma^2_s ds + \int_{\max(t_j,t_l)}^T \sigma^2_s ds}
\end{align*}
with $t_i = 0$ or $t_j=0$ or $t_i = T$ or $t_j=T$ depending on the expression to evaluate.
\end{proof}

Note that the quadratic term of the second-order expansion could have been written in a similar fashion:
\begin{align}
	&\left(\sum_{i=1}^n \hat{\delta}^n_i\left(M_{T}-\frac{M_{T}}{M_{t_i}}\right)+\hat{\delta}^f_i\left(1-\frac{M_{T}}{M_{t_j}}\right) \right)^2 = \left(\sum_{i=1}^n \hat{\delta}^n_i M_{T} +\hat{\delta}^f_i - \left( \hat{\delta}^n_i+ \hat{\delta}^f_i\right)\frac{M_{T}}{M_{t_j}}\right)^2\nonumber\\
	&=\left(\sum_{i=1}^n \hat{\delta}^n_i \right)^2 M_T^2 + \left(\sum_{i=1}^n \hat{\delta}^f_i \right)^2+ \left(\sum_{i=1}^n \hat{\delta}_i \frac{M_{T}}{M_{t_i}}\right)^2\nonumber\\
	&+2\left(\sum_{i=1}^n \hat{\delta}^n_i \right) \left(\sum_{j=1}^n \hat{\delta}^f_j \right) M_T - 2 \left(\sum_{j=1}^n \hat{\delta}^n_j \right) \left(\sum_{i=1}^n \hat{\delta}_i \frac{M_{T}^2}{M_{t_i}}\right)- 2 \left(\sum_{j=1}^n \hat{\delta}^f_j \right) \left(\sum_{i=1}^n \hat{\delta}_i \frac{M_{T}}{M_{t_i}}\right) \text{.}	
	\end{align}

Numerically, in order to speed up the evaluation, the sum over $i,j,l$ can be split as:
\begin{align}
&\sum_{1\leq i,j,l \leq n} \hat{\delta}_i \hat{\delta}_j \hat{\delta}_l \partial_K^3 \mathbb{E}\left[ B_T h\left(F_T e^{\int_{t_i}^T \sigma^2_s ds +\int_{t_j}^T \sigma^2_s ds + \int_{t_l}^T \sigma^2_s ds}-K\right) \right]e^{\int_{\max(t_i,t_j)}^T \sigma^2_s ds + \int_{\max(t_i,t_l)}^T \sigma^2_s ds + \int_{\max(t_j,t_l)}^T \sigma^2_s ds}\nonumber\\
&=6\sum_{i=1}^n \sum_{j=i+1}^n \sum_{l=j+1}^n \hat{\delta}_i \hat{\delta}_j \hat{\delta}_l \partial_K^3 \mathbb{E}\left[ B_T h\left(F_T e^{\int_{t_i}^T \sigma^2_s ds +\int_{t_j}^T \sigma^2_s ds + \int_{t_l}^T \sigma^2_s ds}-K\right) \right]e^{\int_{t_j}^T \sigma^2_s ds + 2 \int_{t_l}^T \sigma^2_s ds}\nonumber\\
&+3\sum_{i=1}^n \sum_{l=i+1}^n \hat{\delta}_i^2 \hat{\delta}_l \partial_K^3 \mathbb{E}\left[ B_T h\left(F_T e^{2\int_{t_i}^T \sigma^2_s ds + \int_{t_l}^T \sigma^2_s ds}-K\right) \right]e^{\int_{t_i}^T \sigma^2_s ds + 2 \int_{t_l}^T \sigma^2_s ds}\nonumber\\
&+3\sum_{i=1}^n \sum_{j=i+1}^n \hat{\delta}_i \hat{\delta}_j^2 \partial_K^3 \mathbb{E}\left[ B_T h\left(F_T e^{\int_{t_i}^T \sigma^2_s ds + 2\int_{t_j}^T \sigma^2_s ds}-K\right) \right]e^{3\int_{t_j}^T \sigma^2_s ds}\nonumber\\
&+\sum_{i=1}^n \hat{\delta}_i^3 \partial_K^3 \mathbb{E}\left[ B_T h\left(F_T e^{3\int_{t_i}^T \sigma^2_s ds}-K\right) \right]e^{3\int_{t_i}^T \sigma^2_s ds}\text{.}
\end{align}

Let $v_t = \sqrt{\int_0^t \sigma^2_s ds}$, applying Theorem \ref{theorem_hl_3} to the vanilla option payoff leads to the additional third-order correction:
\begin{align}
&\frac{1}{6}\left(\sum_{j=1}^n \hat{\delta}^f_i \right)^3 \frac{\phi(d_2)}{\bar{k}^2 v_T}\left(\frac{d_2}{v_T}-1\right)+\frac{1}{6}\left(\sum_{i=1}^n \hat{\delta}^n_i \right)^3 \frac{\bar{k} \phi(d_2)}{\bar{f}^3 v_T}\left(\frac{d_2}{v_T}+2\right)\nonumber\\
&+\frac{1}{2}\left(\sum_{i=1}^n \hat{\delta}^n_i \right) \left(\sum_{j=1}^n \hat{\delta}^f_j \right)^2 \frac{\phi(d_2)}{\bar{f} \bar{k} v_T}\frac{d_2}{v_T}+\frac{1}{2}\left(\sum_{i=1}^n \hat{\delta}^n_i \right)^2 \left(\sum_{j=1}^n \hat{\delta}^f_j \right)\frac{\phi(d_2)}{\bar{f}^2 v_T}\left(\frac{d_2}{v_T}+1\right) \nonumber\\
&-\frac{1}{6}\sum_{1\leq i,j,l \leq n} \hat{\delta}_i \hat{\delta}_j \hat{\delta}_l \frac{\phi\left(d_2 + \frac{3v_T^2-v_{t_i}^2-v_{t_j}^2-v_{t_l}^2}{v_T}\right)e^{3v_T^2-v_{\max(t_i,t_j)}^2-v_{\max(t_i,t_l)}^2-v_{\max(t_j,t_l)}^2}}{\bar{k}^2 v_T}\left(\frac{d_2 + \frac{3v_T^2-v_{t_i}^2-v_{t_j}^2-v_{t_l}^2}{v_T}}{v_T}-1\right)\nonumber\\
&+\frac{1}{2}\left(\sum_{l=1}^n \hat{\delta}^n_l \right) \sum_{1\leq i,j \leq n} \hat{\delta}_i \hat{\delta}_j \frac{\phi\left(d_2 + \frac{2v_T^2-v_{t_i}^2-v_{t_j}^2}{v_T}\right)e^{v_T^2-v_{\max(t_i,t_j)}^2}}{\bar{k} \bar{f} v_T} \frac{d_2 + \frac{2v_T^2-v_{t_i}^2-v_{t_j}^2}{v_T}}{v_T}\nonumber\\
&+\frac{1}{2}\left(\sum_{l=1}^n \hat{\delta}^f_l \right) \sum_{1\leq i,j \leq n} \hat{\delta}_i \hat{\delta}_j 
\frac{\phi\left(d_2 + \frac{2v_T^2-v_{t_i}^2-v_{t_j}^2}{v_T}\right)e^{v_T^2-v_{\max(t_i,t_j)}^2}}{\bar{k}^2 v_T}\left(\frac{d_2 + \frac{2v_T^2-v_{t_i}^2-v_{t_j}^2}{v_T}}{v_T}-1\right)\nonumber\\
&-\frac{1}{2}\left(\sum_{l=1}^n \hat{\delta}^n_l \right)^2 \sum_{i=1}^n \hat{\delta}_i \frac{\phi\left(d_1 - \frac{v_{t_i}^2}{v_T}\right)}{\bar{f}^2 v_T}\left(\frac{d_1 - \frac{v_{t_i}^2}{v_T}}{v_T}+1\right) - \frac{1}{2}\left(\sum_{l=1}^n \hat{\delta}^f_l \right)^2 \sum_{i=1}^n \hat{\delta}_i \frac{\phi\left(d_1 - \frac{v_{t_i}^2}{v_T}\right)}{\bar{k}^2 v_T}\left(\frac{d_1 - \frac{v_{t_i}^2}{v_T}}{v_T}-1\right)\nonumber \\
&-\left(\sum_{j=1}^n \hat{\delta}^n_j \right) \left(\sum_{l=1}^n \hat{\delta}^f_l \right) \sum_{i=1}^n \hat{\delta}_i \frac{\phi\left(d_1 - \frac{v_{t_i}^2}{v_T}\right)}{\bar{k} \bar{f} v_T}\frac{d_1 - \frac{v_{t_i}^2}{v_T}}{v_T}
\end{align}
with $d_1 = d_1(\bar{f},\bar{k}), d_2 = d_2(\bar{f},\bar{k})$.

\section{Numerical experiments}\label{sec:expansion_numerical}

\subsection{Single dividend}
We consider vanilla call options of maturity $T=1.0$ on an asset of spot $S(0)=100$, volatility $\sigma=30\%$, with no drift $r_R = 0\%$ and look at the error in implied Black volatility for a range of strikes when a single dividend of amount $\delta_1 = 7$ is paid at various dates $t_1=\{0.1, 0.5, 0.9\}$. EG-1, EG-2 denote the first and second-order expansions on the displaced strike from \citet{etore2012stochastic}, LF-1, LF-2 denote the first and second-order expansions on the forward, LL-1, LL-2 denote the first and second-order expansions on the Lehman model and HHL is the price obtained via the method of \cite{haug2003back}, which is exact in the case of a single dividend and will serve as a reference.

\begin{figure}[h!]
	\begin{center}  
		\subfigure[\label{fig:pln_lf_1div_0_1_1} $t_1=0.1$]{
			\includegraphics[width=6.5cm]{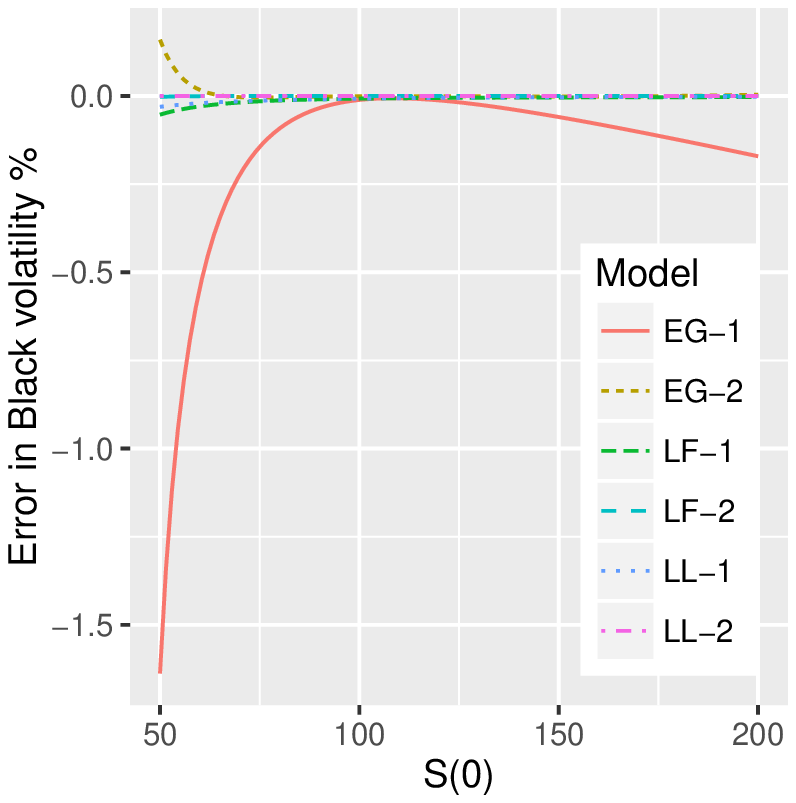}}
		\subfigure[\label{fig:pln_lf_1div_0_1_1e}$t_1=0.1$ removing EG-1]{
			\includegraphics[width=6.5cm]{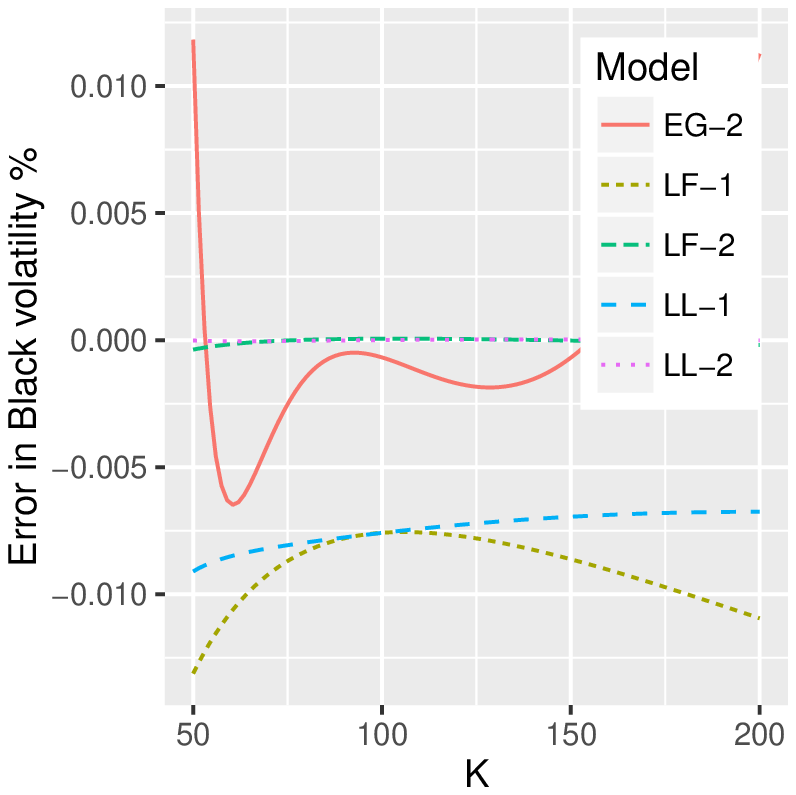}}
	\end{center}
	\caption{Error in implied volatility of first and second-order expansions with $t_1=0.1$.\label{fig:pln_lf_1div_0_1_1a}}
\end{figure}

As expected, the expansions on the forward LF-1 and LF-2 are much more accurate when the dividend ex-date is close to the valuation date. For low strikes, LF-1 becomes more accurate than the second-order expansion on the displaced strike EG-2 (Figures \ref{fig:pln_lf_1div_0_1_1} and \ref{fig:pln_lf_1div_0_1_1e}).

When the dividend is paid at $t_1=0.5$, the expansions on the forward and on the displaced strike have very similar error accross the whole range of strikes (Figure \ref{fig:pln_lf_1div_0_5_1}). When the dividend is paid near maturity, the situation is reversed, and the expansions on the displaced strike are more accurate than the expansions on the forward (Figure\ref{fig:pln_lf_1div_0_9_1}).
\begin{figure}[h]
	\begin{center}  
		\subfigure[\label{fig:pln_lf_1div_0_5_1} $t_1=0.5$]{
			\includegraphics[width=6.5cm]{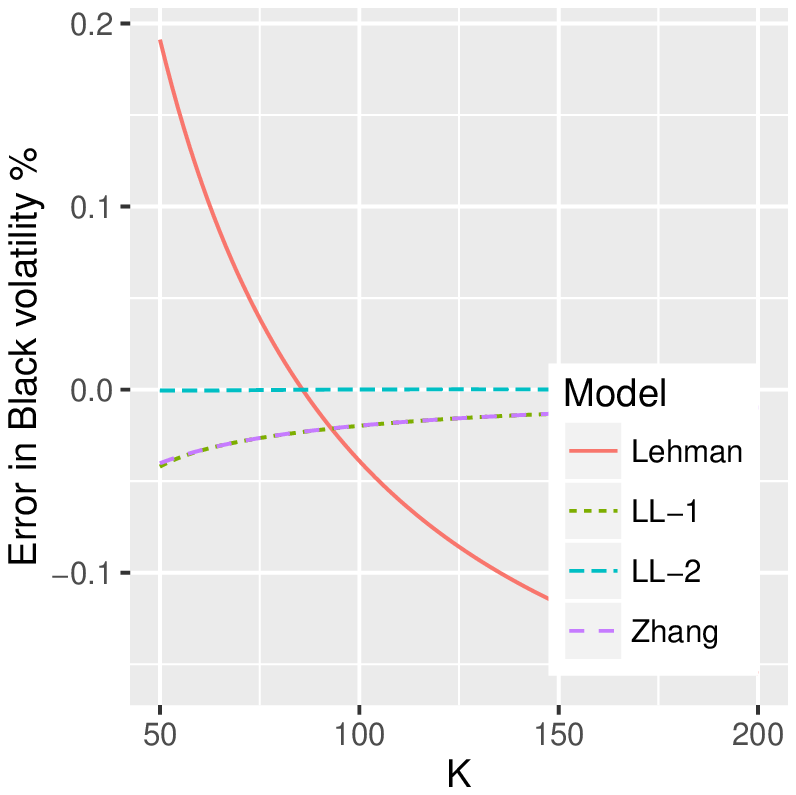}}
		\subfigure[\label{fig:pln_lf_1div_0_9_1}$t_1=0.9$ ]{
			\includegraphics[width=6.5cm]{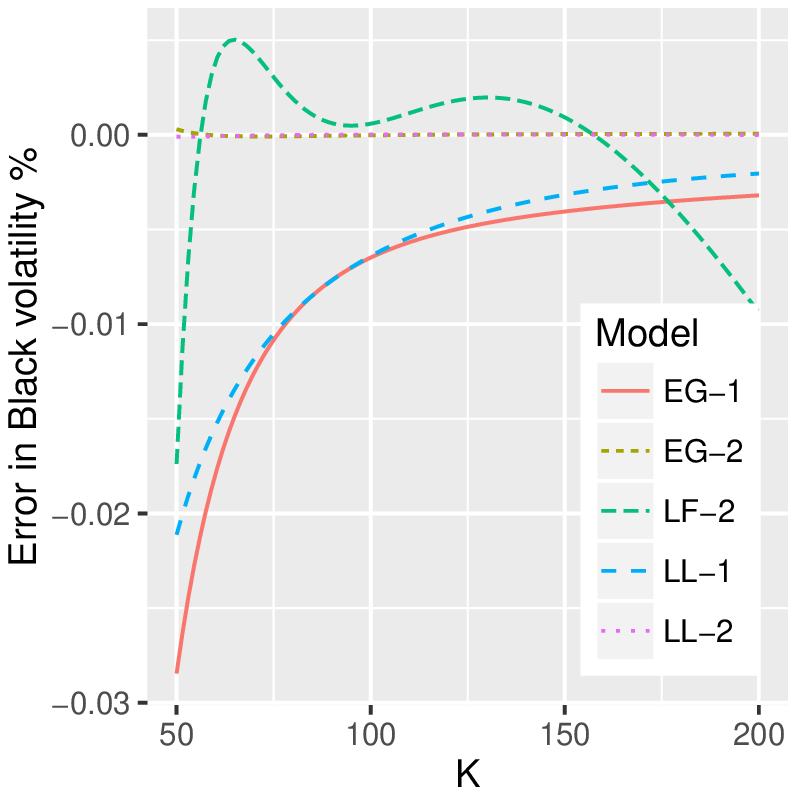}}
	\end{center}
	\caption{Error in implied volatility of first and second-order expansions with various dividend ex-dates.}
\end{figure}

Overall, the expansion using the Lehman model as a proxy is the most accurate, independently whether the strike is far from the asset price or whether the dividend is close to the valuation date or to the maturity date. It is behaves like the best of both forward and strike approximations.
 
In Table \ref{tbl:single}, we compare our expansions to other good approximations, GS denotes the one from \citet{sahel2011matching} and Zhang-1, Zhang-2 denote the first and second-order expansions from \citet{zhang2011fast}.
 \begin{table}[h]
 	\begin{center}
 		\caption{\label{tbl:single_01}Price and absolute error in \% volatility for a single dividend at $t_1 = 0.1$.}
 		\begin{tabular}{c|c|c|c} \hline
 		Strike & 50 & 100 & 150 \\	\hline		
 		HHL    & 43.15079716 ( 0.00e-00)   & 8.42464720 ( 0.00e-00)  & 0.85950498 ( 0.00e-00)  \\
 		Black  & 43.14297574 (-2.43e-01) & 8.33854820 (-2.33e-01)& 0.82943120 (-2.28e-01)  \\
 		Lehman  & 43.15318035 ( 7.24e-02) & 8.41964102 (-1.35e-02)& 0.85328295 (-4.67e-02) \\
 		EG-1    & 43.13065333 (-6.43e-01)& 8.42049085 (-1.12e-02)& 0.84485174 (-1.10e-01) \\
 		LF-1    & 43.15036784 (-1.31e-02)& 8.42184565 (-7.58e-03)& 0.85835513 (-8.62e-03) \\
 		LL-1    & 43.15049919 (-9.10e-03)& 8.42184420 (-7.58e-03)& 0.85857820 (-6.95e-03) \\
 		Zhang   & 43.15050837 (-8.82e-03)& 8.42184920 (-7.57e-03)& 0.85858096 (-6.93e-03) \\
 		EG-2    &43.15118483 ( 1.18e-02)  & 8.42439623 (-6.79e-04)& 0.85941296 (-6.90e-04) \\
 		LF-2    &43.15078502 (-3.71e-04) & 8.42466824 ( 5.69e-05) & 0.85950269 (-1.72e-05) \\
 		LL-2    &43.15079673 (-1.32e-05) & 8.42464789 ( 1.87e-06) & 0.85950819 ( 2.40e-05)  \\
 		Zhang-2 &43.15079600 (-3.56e-05) & 8.42464700 (-5.43e-07)& 0.85950644 ( 1.09e-05)   \\
 		GS      &43.15075570 (-1.27e-03) & 8.42431687 (-8.94e-04)& 0.85940829 (-7.25e-04)  \\
 		EG-3    &43.15089372 ( 2.95e-03)  & 8.42463285 (-3.88e-05)& 0.85952103 ( 1.20e-04)  \\
 		LF-3    &43.15079684 (-9.72e-06) & 8.42464741 ( 5.58e-07) & 0.85950503 ( 3.78e-07)  \\
 		LL-3    &43.15079693 (-6.99e-06) & 8.42464758 ( 1.02e-06) & 0.85950495 (-2.38e-07)  \\ \hline		
 			\end{tabular}
 		\end{center}
 	\end{table}	
 		\begin{table}[h]
 			\begin{center}
 				\caption{\label{tbl:single}Price and absolute error in \% volatility for a single dividend at $t_1 = 0.9$.}
 				\begin{tabular}{c|c|c|c} \hline
 					Strike & 50 & 100 & 150 \\	\hline		
 	HHL    &43.24845580 ( 0.00e-00) &9.07480014 ( 0.00e-00)&1.06252881 ( 0.00e-00) \\
 	Black  &43.14297574 (-2.77e+00) &8.33854820 (-1.99e+00)&0.82943120 (-1.67e+00) \\
 	Lehman &43.25128936 ( 6.32e-02) &9.06953193 (-1.42e-02)&1.05680170 (-3.85e-02) \\
 	EG-1   &43.24718664 (-2.85e-02) &9.07239859 (-6.49e-03)&1.06192565 (-4.05e-03) \\
 	LF-1   &43.22988026 (-4.26e-01) &9.07127538 (-9.52e-03)&1.04735416 (-1.02e-01) \\
 	LL-1   &43.24751323 (-2.11e-02) &9.07242807 (-6.41e-03)&1.06205749 (-3.17e-03) \\
 	Zhang  &43.24752144 (-2.09e-02) &9.07243235 (-6.40e-03)&1.06205987 (-3.15e-03) \\
 	EG-2   &43.24846889 ( 2.93e-04) &9.07479039 (-2.63e-05)&1.06253441 ( 3.76e-05) \\
 	LF-2   &43.24767995 (-1.74e-02) &9.07501725 ( 5.87e-04)&1.06266611 ( 9.22e-04) \\
 	LL-2   &43.24845113 (-1.04e-04) &9.07480497 ( 1.31e-05)&1.06253007 ( 8.45e-06) \\
 	Zhang-2&43.24845254 (-7.30e-05) &9.07480191 ( 4.79e-06)&1.06252951 ( 4.71e-06) \\
 	GS     &43.24865150 ( 4.38e-03) &9.07509909 ( 8.08e-04)&1.06257287 ( 2.96e-04) \\
 	EG-3   &43.24845582 ( 4.35e-07) &9.07480026 ( 3.43e-07)&1.06252875 (-3.84e-07) \\
 	LF-3   &43.24856118 ( 2.36e-03) &9.07478707 (-3.53e-05)&1.06254437 ( 1.05e-04) \\
 	LL-3   &43.24845549 (-6.92e-06) &9.07480027 ( 3.70e-07)&1.06252874 (-4.94e-07) \\ \hline
  		\end{tabular}
 	\end{center}
 \end{table}
The most accurate overall is the second-order approximation from Zhang, followed very closely by the second-order expansion on the Lehman model LL-2. The GS approximation is still very accurate, but its behaviour with a short dated dividend in low strikes is noticeably worse. The LL-2 expansion is a significant improvement over both the expansion around the displaced strike (EG-2) and the expansion around the forward (LF-2).

The third-order expansions increase the accuracy significantly over the second-order counterparts. Figure \ref{fig:pln_lf_1div_0_5_1_o3} shows that the expansions on the forward LF-3 or on the strike EG-3 are not necessarily more accurate than good second-order expansions Zhang-2 or LL-2 for low strikes. The third-order LL-3 is however much more accurate over the full range of strikes, with a maximum error of less than 0.0002 volatility point.
\begin{figure}[h]
	\begin{center}  	
	\includegraphics[width=15cm]{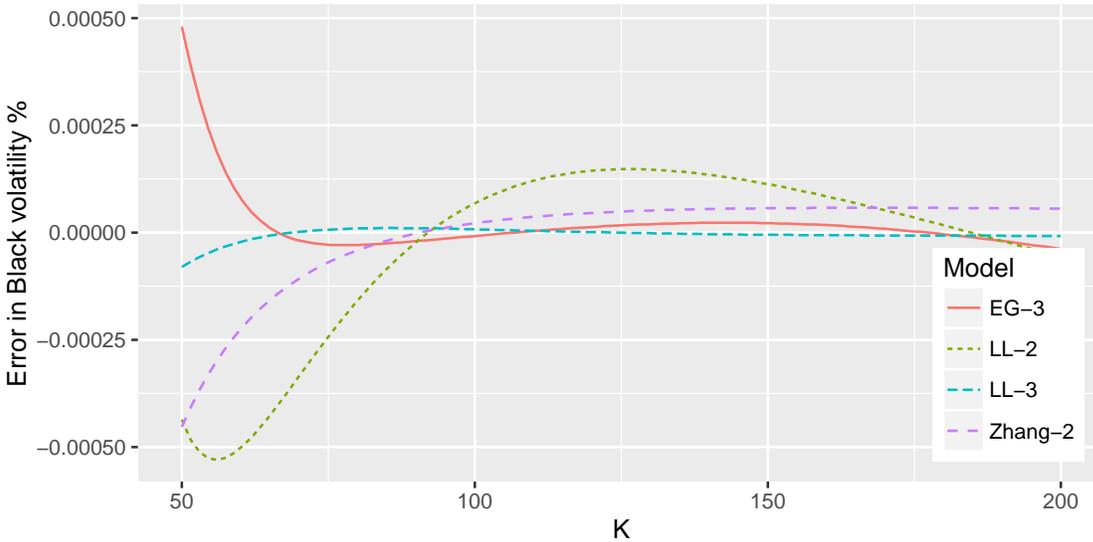}
	\caption{Error in implied volatility of third-order expansions with a single dividend at $t_1=0.5$.\label{fig:pln_lf_1div_0_5_1_o3}}
	\end{center}
\end{figure}
For practical applications however, the second-order expansions LL-2 is sufficiently accurate as its maximum error in implied volatility is around 0.001 volatility point for a dividend at $t_1=0.5$.

\subsection{Many dividends}
We reuse here the example of \citet{sahel2011matching} and consider the underlying asset price $S(0)=100$, the volatility $\sigma=25\%$, the discount and repo rate $r=3\%$, the maturity $T=10$ and dividends of amount $\delta_i=2$ at dates $t_i = 0.5i+ \frac{1}{365}$. The first dividend happens just one day after the valuation date.

On this example, the HHL method from \citet{haug2003back} is not exact anymore as we use their shifted lognormal approximation to work around the multiple integral problem. In fact, in more extreme scenarios, the HHL method can lead to implausible prices: prices that can not be inverted via the Black formula or prices below the intrinsic value. 

The reference value is obtained through a finite difference method, with a fine enough grid. We will ignore the simpler expansions LF-1, EG-1 as they are not practical even with a single dividend.

The HHL method is less accurate than our first-order expansion (Figure \ref{fig:pln_20div_gocsei_order1}) and first-order Zhang approximation has the same error as our first-order expansion LL-1. Overall, the second-order methods of \citet{zhang2011fast} and \citet{sahel2011matching} are the most accurate. Our second-order expansion LL-2 is nearly as good, and improves greatly over the second-order expansion on the displaced strike.
\begin{figure}[htb]
	\begin{center}  
		\subfigure[\label{fig:pln_20div_gocsei_order1} first-order methods]{
			\includegraphics[width=7cm]{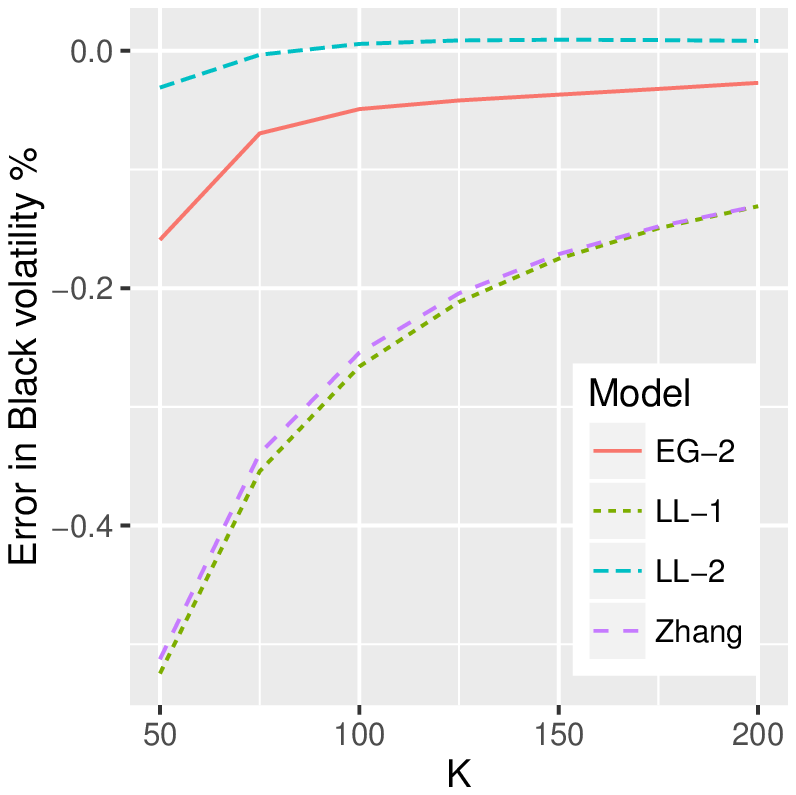}}
		\subfigure[\label{fig:pln_20div_gocsei} second-order methods ]{
			\includegraphics[width=7cm]{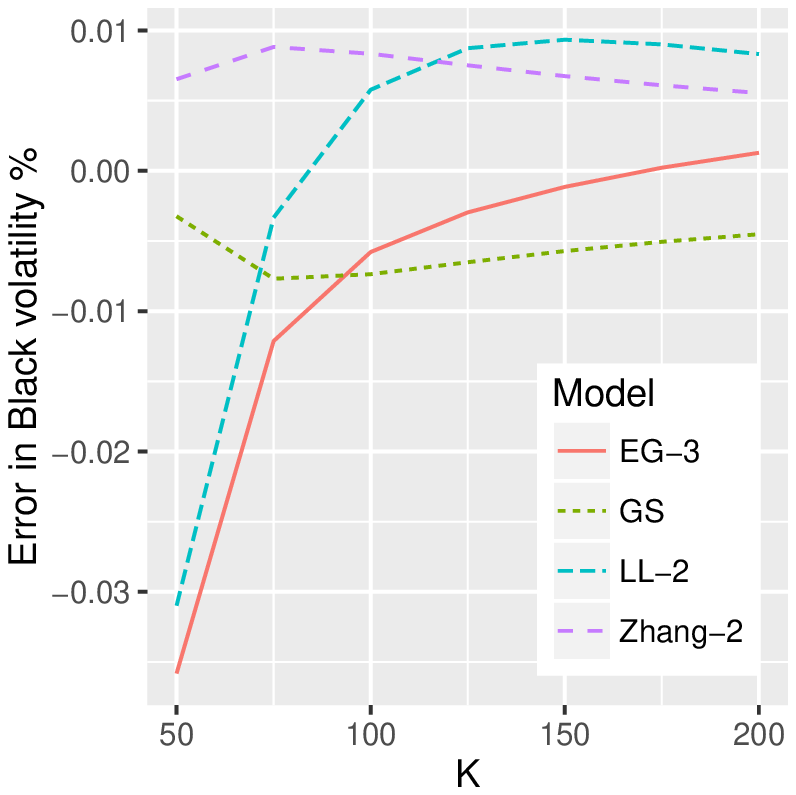}}
		\subfigure[\label{fig:pln_20div_gocsei_order3} third-order against second-order ]{
			\includegraphics[width=14cm]{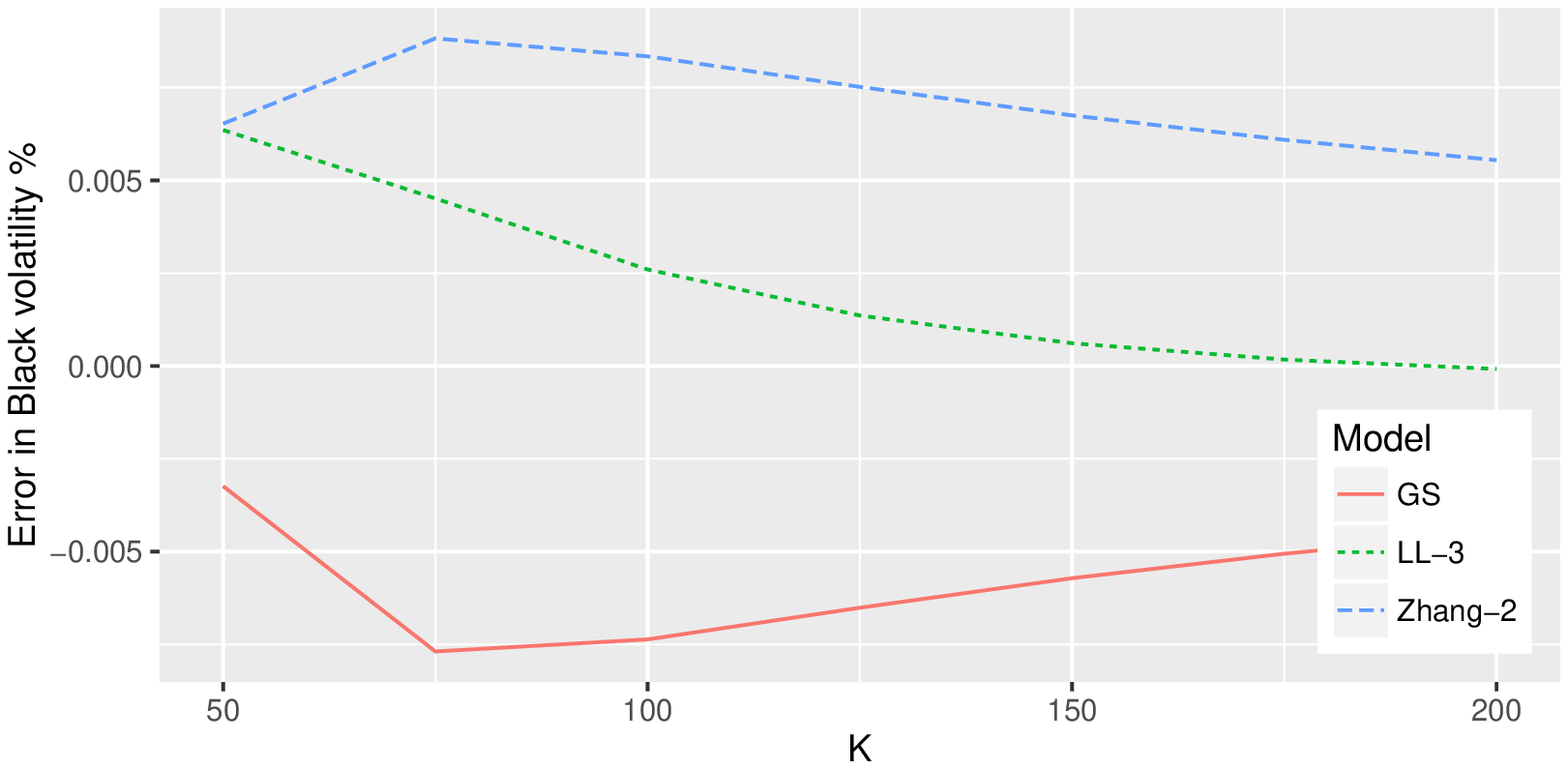}}
				\caption{Error in implied volatility on Gocsei et al. semi-annual dividend case of ten years maturity.}
	\end{center}
\end{figure}
The third-order expansion on the displaced strike EG-3 does not really improve over LL-2 or Zhang-2 expansions on this example. But the third-order expansion on the Lehman model LL-3 is significantly more accurate.

On another example of \citet{sahel2011matching} where they consider a long term option on an index paying weekly dividends with $S(0)=3000, r=3\%, \delta_i= 2, T=20$, the GS approximation is noticeably worse than Zhang-2 for low strikes, in line with the behaviour we noticed in the single dividend case (Figure \ref{fig:pln_20y_gocsei_o3}). This is also an example where the HHL method gives a wrong price for low strikes.
\begin{figure}[htb]
	\begin{center}  
		\subfigure[Error in implied volatility of an option of maturity 20 years with weekly dividends.\label{fig:pln_20y_gocsei_o3}]{
			\includegraphics[width=7cm]{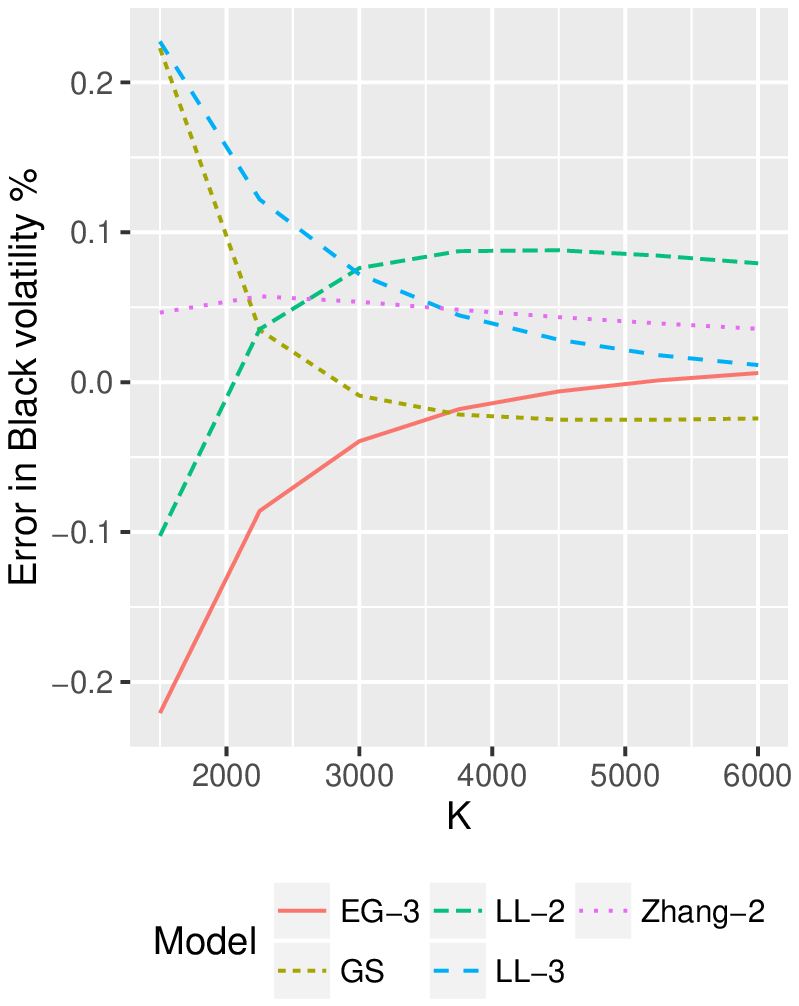}}
		\subfigure[Error in price on the extreme example of an option of maturity 1 year with strike $K=50$ and  semi annual dividends $\delta_i=25$ with $\sigma=80\%$.\label{fig:pln_zhang_extreme} ]{
			\includegraphics[width=7cm]{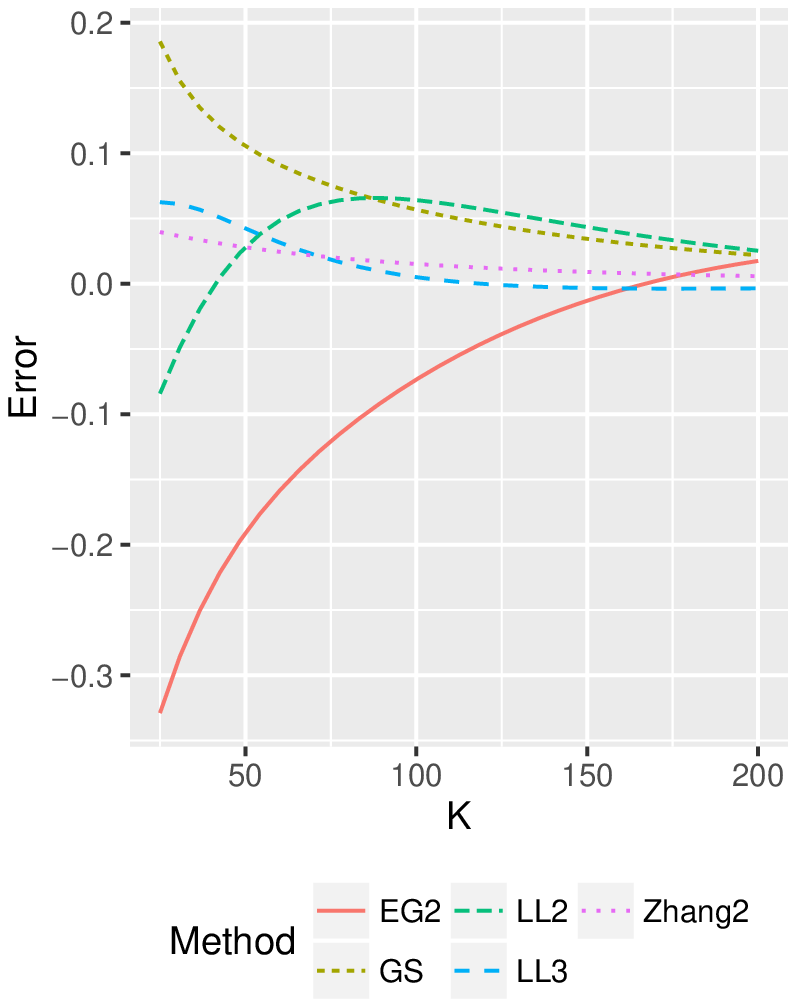}}
		\caption{Challenging examples for third-order methods.}
	\end{center}
\end{figure}
The third-order expansion LL-3 does not really improve on the second-order expansion LL-2. Note that this is not a particularly practical example, as, in reality, long term dividends, typically beyond four year, will be modelled as proportional dividends and not as cash dividends. Interestingly, if we consider a maturity of four years instead of twenty years, with otherwise the same asset and the same weekly dividends, LL-3 is then much more accurate, very much like the example of Figure \ref{fig:pln_20div_gocsei_order3}.

It is interesting to look at a more extreme example like the one from \citet{zhang2011fast}. The author considers an option of maturity one year and strike $K=50$ on an asset with price $S(0)=100$ and very large semi annual dividends $\delta_i=25$ at dates $t_0=0.3$ and $t_1=0.7$ with a very high volatility $\sigma=80\%$. On this example, the second-order method of Zhang has the best overall accuracy. Our third-order expansion on the Lehman model, while more accurate than the expansion on the displaced strike EG-3 is only as accurate as Zhang, while the method from \citet{sahel2011matching} results in an error more than three times larger (Figure \ref{fig:pln_zhang_extreme}).

In terms of performance, the second-order methods are relatively similar  as they all involve $O(n^2)$ exponential evaluations (Table \ref{tbl:zhang_gobet_perf}). Zhang-2 is the slowest as it involves some numerical solving, even if two to three steps of Halley's method are generally enough.

\begin{table}[ht]
	\caption{\label{tbl:zhang_gobet_perf}Time in seconds taken to price 1000 European options. FDM denotes a finite difference method discretisation of the spot PDE with 500 space steps and 100 time steps, with additional time-steps at the dividend dates.}
	\begin{center}
	\begin{tabular}{c r r r }\hline
		Number of dividends & 1 & 10 & 100 \\ \hline
		Zhang-1 & 0.0010 & 0.0034 & 0.038 \\
		Zhang-2 & 0.0016 & 0.0100 & 0.219\\
		EG-2 & 0.0006 & 0.0046 & 0.255\\
		LL-2 & 0.0006 & 0.0056 & 0.260 \\
		GS & 0.0007 & 0.0053 & 0.220 \\
		EG-3 & 0.0007 &  0.0151 & 7.240 \\
		LL-3 & 0.0007 &  0.0152 & 7.360 \\
		HHL & 0.0902 & 0.918 & 9.090 \\
		FDM & 2.4400 & 3.340 & 9.910\\ \hline %2.44 & 3.34 & 9.91 %58.1 for 1000		
	\end{tabular}
	\end{center}
\end{table}
With a single dividend, the second-order approximations can be 5000 times faster than a finite difference method, with 100 dividends, they are still around forty times faster than a finite difference method. With around 1000 dividends they become of similar performance. For third-order approximations, the threshold is around 100 dividends. With 1000 dividends, they become not so practical as it takes then around 7 seconds to price one option.

\section{Conclusion}
Among the first-order expansions, our expansion with the Lehman proxy is as robust as Zhang first-order expansion, but purely analytical, while the latter expansion requires a (simple) numerical solving.

Among the second-order methods, we found the Zhang method to be the most accurate across a wide range of strikes, dividend dates and dividend amounts. Our second-order expansion using the Lehman model as a proxy is as robust as Zhang second-order formula across strikes and expiries, even with unrealistically large dividends, but it has a larger error. It is however possible to compute the Greeks through analytical formulas even in the second and third-order expansions, and the formulas for the Greeks are not more complex than the formula for the price. The second-order Zhang method does not allow this as easily.

Finally, our third-order expansion using the Lehman model as proxy is the most accurate in practical use cases. Its accuracy is not as good for very long term options where dividends would be modelled as cash. Furthermore, it will be significantly slower for a very large number of dividends. In reality however, the practice is to model long term dividends as proportional and those drawbacks may not be relevant.

A classical approach to price American option under the piecewise-lognormal model consists in applying a finite difference method to the corresponding partial differential equation. The accurate and fast expansions for European options prices under the piecewise-lognormal model allow for an interesting alternative method to price American options. Indeed, %Vellekoop and Nieuwenhuis
 \citet{vellekoop2011integral} show that exercise boundary of an American option, and thus its price, may be computed by solving a non-linear integral equation. The one-dimensional integral involves the cumulative probability density of the piecewise-lognormal model, which is generally unknown, but may easily be approximated using the expansions presented in this paper.

\bibliographystyle{rAMF}
\bibliography{expansion_cash_dividend}

\end{document}